\documentclass[twoside,12pt]{article}

\usepackage[utf8]{inputenc}
\usepackage[T1]{fontenc}
\usepackage{enumitem}
\usepackage{textcomp}
\usepackage{relsize}
\usepackage{setspace}
\usepackage[fleqn]{amsmath}
\usepackage{amssymb}
\usepackage{amsthm}
\usepackage{calrsfs}
\usepackage{mathtools}
\usepackage{lmodern}
\usepackage[a4paper,left=2cm,right=2cm,top=2cm,bottom=3cm,footskip=1.5cm,heightrounded]{geometry}
\usepackage{graphicx}
\usepackage[margin=2cm]{caption}
\usepackage[usenames, dvipsnames]{xcolor}
\usepackage{microtype}
\usepackage{hyperref}
\usepackage{url}
\usepackage{titling}
\usepackage{caption}
\hypersetup{pdfstartview=XYZ}
\usepackage{float}
\usepackage{authblk}

\newcommand{\N}{\mathbb{N}}
\renewcommand{\P}{\mathcal{P}}
\newcommand{\Z}{\mathbb{Z}}
\newcommand{\defeq}{\vcentcolon=}
\newcommand{\eqdef}{=\vcentcolon}

\newcommand{\Adj}{\textup{Adj}}
\newcommand{\Span}{\textup{span}}
\newcommand{\dist}{\textup{dist}}
\newcommand{\Perim}{\textup{Perim}}

\newcommand{\nom}[2]{#1 \textsc{#2}}
\newcommand{\article}[1]{"#1"}
\newcommand{\ouvrage}[1]{\textit{#1}}

\newtheoremstyle{the}{20pt}{20pt}{\it}{}{\bfseries}{.}{ }{}
\newtheoremstyle{def}{20pt}{20pt}{}{}{\bfseries}{.}{ }{}
\newtheoremstyle{rem}{20pt}{20pt}{}{}{\it}{.}{ }{}

{\theoremstyle{the} \newtheorem{theoreme}{Theorem}[section]}
{\theoremstyle{the} \newtheorem{lemme}[theoreme]{Lemma}}
{\theoremstyle{the} \newtheorem{proposition}[theoreme]{Proposition}}
{\theoremstyle{the} \newtheorem{corollaire}[theoreme]{Corollary}}
{\theoremstyle{def} \newtheorem{definition}[theoreme]{Definition}}
{\theoremstyle{def} \newtheorem{notation}[theoreme]{Notation}}
{\theoremstyle{rem} }
{\theoremstyle{rem} }

\title{An update on the coin-moving game on the square grid}
\author[1,3]{\nom{Florian}{Galliot}}
\author[1,3]{\nom{Sylvain}{Gravier}}
\author[2,3]{\nom{Isabelle}{Sivignon}}

\affil[1]{Univ. Grenoble Alpes, CNRS, Institut Fourier, 38000 Grenoble, France}
\affil[2]{Univ. Grenoble Alpes, CNRS, Grenoble INP, GIPSA-lab, 38000 Grenoble, France}
\affil[3]{Univ. Grenoble Alpes, Maths à Modeler, 38000 Grenoble, France}

\date{}

\begin{document}

\maketitle

\captionsetup{width=13.3cm,labelsep=period,labelfont={bf}}

\begin{abstract}
	\noindent This paper extends the work started in 2002 by Demaine, Demaine and Verill (DDV) on coin-moving puzzles. These puzzles have a long history in the recreational literature, but were first systematically analyzed by DDV, who gave a full characterization of the solvable puzzles on the triangular grid and a partial characterization of the solvable puzzles on the square grid. This article specifically extends the study of the game on the square grid.
	\\ Notably, DDV’s result on puzzles with two “extra coins” is shown to be overly broad: this paper provides counterexamples as well as a revised version of this theorem. A new method for solving puzzles with two extra coins is then presented, which covers some cases where the aforementioned theorem does not apply.
	\\ Puzzles with just one extra coin seem even more complicated, and are only touched upon by DDV. This paper delves deeper, studying a class of such puzzles that may be considered equivalent to a game of "poking" coins. Within this class, some cases are considered that are amenable to analysis.

\end{abstract}

\vspace{.5\baselineskip}
\section*{Introduction}
\vspace{1\baselineskip}

\hphantom{\indent}We study a one-player game that is played on an undirected simple graph $G=(V,E)$, referred to as the \textit{game graph} or \textit{board}. Throughout the game, there will be coins sitting on some of the vertices of $G$ (at most one per vertex). The coins are indistinguishable and define a \textit{configuration}, i.e. a finite subset $C \subseteq V$ where we see each element of $C$ as a coin sitting on the corresponding vertex. A legal \textit{move} consists in moving a single coin to a free vertex so that, after the move, that coin has at least two other coins adjacent to it. This is called the \textit{2-adjacency} restriction. Given two configurations $A$ and $B$, we want to know whether the \textit{puzzle} $A \xrightarrow{?} B$ is solvable: starting from $A$, is it possible to reach $B$ using only legal moves? In the positive case, we would like an explicit winning sequence of moves. This game falls into the category of reconfiguration problems on graphs.
\\ \indent Instances of this game, or rather a variation with tightly packed coins that can only be slid in the plane without collision, appear in the literature as early as the 1950s in \cite{Lan51} and \cite{Lan53}. Figure \ref{Example_puzzles} features a couple of classic puzzles on the triangular grid as well as a more rare puzzle on the square grid. Such examples also appear in \cite{Gar75} and \cite{BCG82} among others, but it is not until 2002 that general puzzles with these rules have been studied, in \cite{DDV02} which serves as foundation for the present paper. The authors give a full characterization for solvable puzzles on the triangular grid (up to a minor omission which is easily settled: see \cite{Gal19}). Furthermore, they address a large family of puzzles on the square grid, providing polynomial time solving algorithms. Another coin-moving game on the square grid, with different rules but similar methods and also polynomial time algorithms in many cases, is studied in \cite{DP06}.

\begin{figure}[h]
	\centering
	\includegraphics[scale=.35]{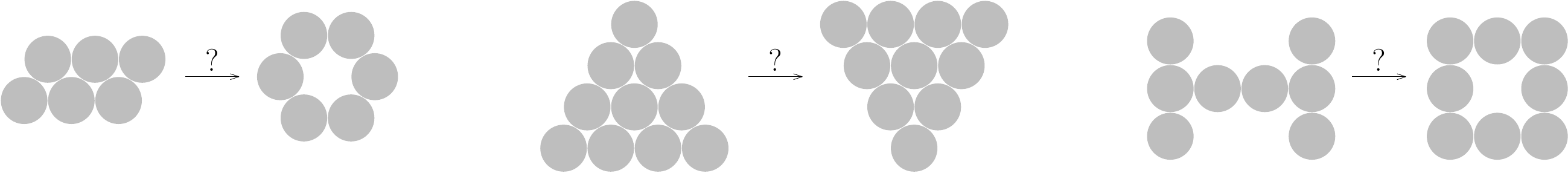}
	\caption{The game graph is the triangular grid in the left and middle puzzles, and the square grid in the right puzzle. The left (resp. middle, resp. right) puzzle is solvable in 2 (resp. 3, resp. 4) moves.}\label{Example_puzzles}
\end{figure}

\indent In this paper, we exclusively study the case where $G$ is the square grid (each vertex has four neighbors: left, right, top, bottom). A first natural question is : what can we reach starting from a configuration $A$? A central observation is that, during the moves, all coins remain inside of the \textit{span} of $A$, which is a finite set (union of rectangles) obtained from $A$ by including all vertices that have at least two neighbors in $A$ and iterating this process until no more vertex can be included. Any configuration $B$ that we wish to obtain from $A$ must therefore satisfy $B \subseteq \Span(A)$, which implies $\Span(B) \subseteq \Span(A)$. In the study that is made in \cite{DDV02}, a key information is the number of coins that can be removed from $A$ while maintaining a span containing that of $B$: we call them \textit{extra coins in $A$ relatively to $B$}. The more extra coins at our disposal, the more flexibility with respect to the span constraint, hence, the easier a puzzle. Section \ref{Section1} goes over fundamental observations and notions introduced in \cite{DDV02}, including that of span and extra coins, with the addition of some basic results of our own that will be used in this paper. Section \ref{Section2} then recalls the method that is employed to solve puzzles in \cite{DDV02}, which we will also use ourselves.
\\ \indent In \cite{DDV02}, the authors claim that, up to an additional condition on $B$, two extra coins in $A$ relatively to $B$ are enough to solve any puzzle. Unfortunately, the proof only works if $A$ and $B$ have the same span. In Section \ref{Section3}, we extend this result to the case where each connected component of $\Span(A)$ contains at most one connected component of $\Span(B)$.
\\ \indent If this latter condition is not satisfied however, then Section \ref{Section4} provides, for any $n$, an example of an unsolvable puzzle where the span of $A$ is an $n \times n$ square and $A$ has about $\frac{n}{2}$ extra coins relatively to $B$ (with $B$ satisfying the additional condition from \cite{DDV02}, and some more). It is worth noticing that $n$, in that case, is also the minimum cardinality of a configuration having an $n\times n$ span. We then give a new sufficient condition for a puzzle to be solvable, which shows in particular that this number $\frac{n}{2}$ is somehow tight.
\\ \indent As mentioned in \cite{DDV02}, the case of a single extra coin is even more complicated. The first non-trivial results for this case are presented in Section \ref{Section5}, under the restriction that $A$ and $B$ have the same span and contain exactly one coin (the extra coin) more than the minimum possible for their span. We show that our game then reduces to a \textit{poking game} where, instead of the 2-adjacency rule, a coin can be slid onto a neighboring vertex under certain conditions. We give necessary and sufficient conditions for poking puzzles involving what is called a chain of coins.

\vspace{1\baselineskip}
\section{Preliminaries}\label{Section1}
\vspace{1\baselineskip}

\subsection{Notations and first observations}

We use the (self-)dual grid for graphical representations of the game: each vertex, or \textit{position}, is seen as a square and coins are placed at the center of squares.

\begin{figure}[h]
	\centering
	\includegraphics[scale=.5]{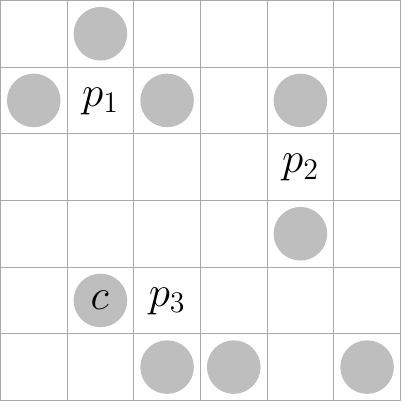}
	\caption{A configuration $C$. From $C$, a possible move would be $c \mapsto p_1$ or $c \mapsto p_2$ for example. However, $c$ cannot be moved to $p_3$, because that position only has one neighboring coin other than $c$.}
\end{figure}

\begin{notation} We introduce the following notations:
	\begin{itemize}[noitemsep,nolistsep]
		\item Moving a coin $c$ to a valid destination $p$ is denoted by $c \mapsto p$. Note that a coin $c$ is simply an element of $V$ representing its location (since the coins are indistinguishable, a coin $c$ is nothing but an occupied position $p$). Typically, the notation $c$ is used to designate an occupied position, and the notation $p$ is used for a position that is either unoccupied or not necessarily occupied.
		\item If there is a single move $c \mapsto p$ from $A$ to $B$, we write $A \xmapsto{c \,\mapsto p} B$ or simply $A \mapsto B$.
		\item If there exists a sequence of moves from $A$ to $B$, we write $A \to B$, otherwise we write $A \not\to B$. By convention, an empty sequence of moves is allowed so that $A \to A$.
		\item If $A \to B$ and $B \to A$ then we may write $A \leftrightarrow B$.
	\end{itemize}
\end{notation}

\begin{proposition}\label{trivial1}
	Let $A \neq B$ be configurations. If $A \to B$ then there exists a coin in $B$ that has at least two neighboring coins in $B$.
\end{proposition}

\begin{proof}
	The last moved coin has at least two neighboring coins by the 2-adjacency rule.
\end{proof}

\begin{proposition}\label{trivial2}
	Let $A \neq B$ be configurations. If there exists a coin $b \in B$ such that $B \setminus \{b\}$ consists of all isolated coins, then $A \to B$ if and only if $A \mapsto B$.
\end{proposition}

\begin{proof}
	Suppose that $A=A_0 \xmapsto{c_1 \,\mapsto p_1} A_1 \xmapsto{c_2 \,\mapsto p_2} \ldots \xmapsto{c_T \,\mapsto p_T} A_T = B$ and that $B \setminus \{b\}$ consists of all isolated coins. In particular all coins in $B \setminus \{b\}$ have at most one neighboring coin in $B$, so $p_T=b$ by the 2-adjacency rule. Since $A_{T-1} \setminus \{c_T\} = B \setminus \{p_T\}$, this means $A_{T-1} \setminus \{c_T\}$ consists of all isolated coins. In particular all coins in $A_{T-1} \setminus \{c_T\}$ have at most one neighboring coin in $A_{T-1}$, so $p_{T-1}=c_T$ by the 2-adjacency rule. Continuing so, we get $p_{T-1}=c_T, p_{T-2}=c_{T-1},\ldots,p_1=c_2$. In conclusion, we have moved the same coin each time, so we could have moved it just once instead and got $A \xmapsto{c_1 \,\mapsto p_T} B$.
\end{proof}

\subsection{Picking up and dropping coins}

\begin{definition}{\textup{\cite{DDV02}}}
	Consider the following actions:
	\begin{itemize}[noitemsep,nolistsep]
		\item \textit{Pick up} a coin: remove a coin from the board, without any restriction on its position.
		\item \textit{Drop} a coin: put a previously picked up coin back on the board, with the 2-adjacency restriction.
	\end{itemize}
	Picked up coins that have not yet been dropped may be referred to as \textit{coins in hand}.
\end{definition}

\noindent It is shown in \cite{DDV02} that the game is unchanged if the player is allowed to pick up and drop coins additionally to moving them:
\begin{proposition}\label{alternative}{\textup{\cite{DDV02}}}
	A puzzle is solvable (by moving coins) if and only if it is solvable by moving, picking up and dropping coins.
\end{proposition}

\begin{notation}
	In this variation, the state of the game at any given moment is described by the configuration $A$ on board and the number $k$ of coins in hand: we denote this information by $A^{+k}$. For example, the notation $A^{+k} \to B^{+k'}$ means that, from the configuration $A$ with $k$ coins in hand at the start, it is possible (via moving, picking up and dropping coins) to reach the configuration $B$ with $k'=k+|A|-|B|$ coins in hand at the end.
\end{notation}
	
\subsection{Span of a configuration}

As mentioned in the introduction, the notion of span is central to the study of the game on the square grid.

\begin{notation}
	Let $C$ be a configuration, we denote by $\Adj(C) \subseteq V \setminus C$ the set of all positions outside $C$ that have at least two neighbors in $C$.
\end{notation}

\begin{definition}{\textup{\cite{DDV02}}}\label{Definition_span}
	The \textit{span} of a configuration $C$, denoted by $\Span(C)$, is the limit of the non-decreasing sequence of configurations $(C_i)_{i \geq 0}$ defined recursively by $C_0=C$ and $C_{i+1}=C_i \cup \Adj(C_i)$. In other words, $\Span(C)$ is the set of all positions that could be reached from $C$ if we had unlimited coins to add to the board at successive positions satisfying the 2-adjacency rule.
\end{definition}

\begin{proposition}\label{span_necessary}{\textup{\cite{DDV02}}}
	The span never increases during moves: if $\, A \to B \,$ then $\,\Span(A) \supseteq \Span(B)$.
\end{proposition}

\begin{definition}
	Identifying the square grid as $\Z^2$, an $m \times n$ \textit{rectangle} $R$ is a set of positions of the form $I \times J$ where $I$ and $J$ are intervals of cardinality $m$ and $n$ respectively, so that each \textit{row} of $R$ contains $m$ positions and each \textit{column} of $R$ contains $n$ positions. We say $R$ is \textit{even} (resp. \textit{odd}) if its \textit{half-perimeter} $m+n$ is even (resp. odd).
\end{definition}

\begin{proposition}\label{span_carre}{\textup{\cite{DDV02}}}
	The span of any configuration $C$ is a union of rectangles at distance at least 3 from each other (these rectangles are called the components of $\Span(C)$).
\end{proposition}

\noindent An example can be seen in Figure \ref{Example_span}. We add the following elementary property of the span which will be useful later.

\begin{figure}[h]
	\centering
	\includegraphics[scale=.4]{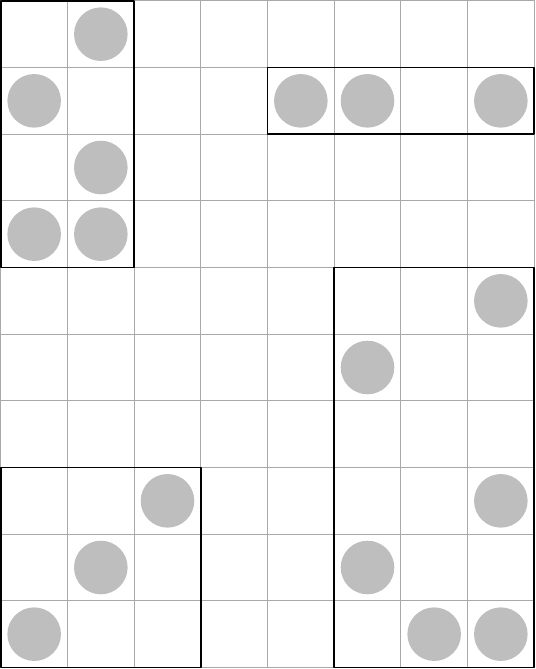}
	\caption{The span of a configuration. In this example, there are four components.}\label{Example_span}
\end{figure}

\begin{proposition}\label{span_carre2}
	Let $C$ be a configuration and let $R$ be a component of $\Span(C)$. Then $C$ contains at least one coin in each of the following: the top row of $R$, the bottom row of $R$, the leftmost column of $R$, the rightmost column of $R$, any union of two consecutive rows in $R$, any union of two consecutive columns in $R$.
\end{proposition}

\begin{proof}
	Using the symmetries, we only address the case of the top row and the union of two consecutive rows. Let $C=C_0,C_1,\ldots,C_s=\Span(C)$ as in Definition \ref{Definition_span} (indeed, the sequence $(C_i)_{i \geq 0}$ is stationary by finiteness of the span).
	\begin{itemize}[noitemsep,nolistsep]
		\item Suppose for a contradiction that the top row $R'$ contains no coin in $C$ i.e. $C \cap R' = \varnothing$. Let $i \in \{0,\ldots,s\}$ be smallest such that $C_i \cap R' \neq \varnothing$: we have $i \geq 1$ and $C_{i-1} \cap R' = \varnothing$. Obviously, any position in $R'$ has at most one neighbor in $R \setminus R'$, so $\Adj(C_{i-1})\cap R' = \varnothing$. Since $C_i=C_{i-1} \cup \Adj(C_{i-1})$, we get $C_i \cap R' = \varnothing$ which is a contradiction.
		\item The previous proof still works if we replace $R'$ by any union of two consecutive rows, since the key argument that any position in $R'$ has at most one neighbor in $R \setminus R'$ still holds. \qedhere
	\end{itemize}
	\renewcommand{\qedsymbol}{}
\end{proof}

\subsection{Minimal/minimum configurations}

\begin{definition}
	Let $C$ be a configuration.
	\begin{itemize}[noitemsep,nolistsep]
		\item We say $C$ is \textit{minimal}, if the removal of any coin in $C$ decreases the span.
		\item We say $C$ is \textit{minimum} if there is no configuration $C'$ with same span as $C$ such that $|C'|<|C|$ (in particular, $C$ is then minimal).
	\end{itemize}
\end{definition}

\begin{proposition}\label{minimal}{\textup{\cite{DDV02}}}
	Any move played from a minimal configuration decreases the span.
\end{proposition}

\noindent As noticed in \cite{DDV02}, the cardinality of minimum configurations (as well as a lot of information on their structure) is well known thanks to the following classic problem from folklore. In a rectangular parcel $R$ consisting of small squares arranged in a grid, some squares are initially invaded by weeds. Time passes, and at each time step, any square that is adjacent to at least two weeds-covered squares gets invaded in turn. How many squares need to be covered initially for the entire parcel to be invaded in the end? Since the rule for the propagation of the weeds is exactly the same as for the construction of the span, the answer coincides with the cardinality of a minimum configuration with span $R$. This problem was first published in \cite{Kva86} for a $10 \times 10$ parcel. An elegant solution is obtained via an invariant which is the perimeter of the invaded area:
\begin{proposition}\label{prop_minimum}
	Let $R$ be an $m \times n$ rectangle and let $M$ be a minimum configuration with span $R$. We have $|M|=\left\lceil\frac{m+n}{2}\right\rceil$. Moreover:
	\begin{itemize}[noitemsep,nolistsep]
		\item If $R$ is even, then all coins in $M$ are isolated.
		\item If $R$ is odd, then all coins in $M$ are isolated except possibly for a single pair of adjacent coins.
	\end{itemize}
\end{proposition}

\begin{proof}
	Define $\Perim(C)$ as the perimeter of the union of unit squares on board corresponding to all occupied positions in a configuration $C$. For any configuration $C$ and any $p \not\in C$, we have $\Perim(C \cup \{p\})=\Perim(C)+4-2k$ where $k$ is the number of coins in $C$ that are a neighbor of $p$: indeed, in terms of the perimeter, each side of $p$ removes 1 if it neighbors a coin in $C$ or adds 1 if it does not. Therefore, any configuration $C$ satisfies $\Perim(C)=4|C|-2l$ where $l$ is the number of pairs of adjacent coins in $C$. Moreover, the perimeter cannot increase during the construction of the span, since each added coin has at least two neighbors amongst the coins that are already present: this implies that any configuration $C$ with span $R$ satisfies $2(m+n)=\Perim(\Span(C))\leq \Perim(C) \leq 4|C|$ hence $|C|\geq \left\lceil\frac{m+n}{2}\right\rceil$. Conversely, it is easy to find configurations with span $R$ that have exactly $\left\lceil\frac{m+n}{2}\right\rceil$ coins (think of a diagonal of coins if $m=n$, or more generally an ‘L’ shape as in Definition \ref{Definition_L}), so the cardinality of minimum configurations with span $R$ is exactly $\left\lceil\frac{m+n}{2}\right\rceil$. Finally, let $M$ be a minimum configuration with span $R$ and let $l$ be the number of pairs of adjacent coins in $M$, we have $2(m+n) \leq \Perim(M)=4\left\lceil\frac{m+n}{2}\right\rceil-2l$, so in conclusion:
	\begin{itemize}[noitemsep,nolistsep]
		\item If $R$ is even, then $2(m+n) \leq 4\frac{m+n}{2}-2l$ hence $l=0$ i.e. all coins in $M$ are isolated.
		\item If $R$ is odd, then $2(m+n) \leq 4\frac{m+n+1}{2}-2l$ hence $l \in \{0,1\}$ i.e. all coins in $M$ are isolated except possibly for a single pair of adjacent coins. \qedhere
	\end{itemize}
	\renewcommand{\qedsymbol}{}
\end{proof}

\subsection{Extra coins and redundant coins}

\noindent We recall the notion of extra coins defined in the introduction:
\begin{definition}\label{def_extra}
	Let $A$ be a configuration and $k \in \N$.
	\begin{itemize}[noitemsep,nolistsep]
		\item A \textit{set of extra coins in $A$} is a subset $A' \subset A$ such that $\Span(A \setminus A')=\Span(A)$. We say $A$ \textit{has $k$ extra coins} if it contains a set of $k$ extra coins. For example, $A$ has one extra coin if and only if $A$ is not minimal.
		\item Let $B$ be a configuration. A \textit{set of extra coins in $A$ relatively to $B$} is a subset $A' \subset A$ such that $\Span(A \setminus A') \supseteq \Span(B)$. We say $A$ \textit{has $k$ extra coins relatively to $B$} if it contains a set of $k$ extra coins relatively to $B$. If $A$ and $B$ have same span, then this definition coincides with the previous one.
	\end{itemize}
\end{definition}

\begin{figure}[h]
	\centering
	\includegraphics[scale=.5]{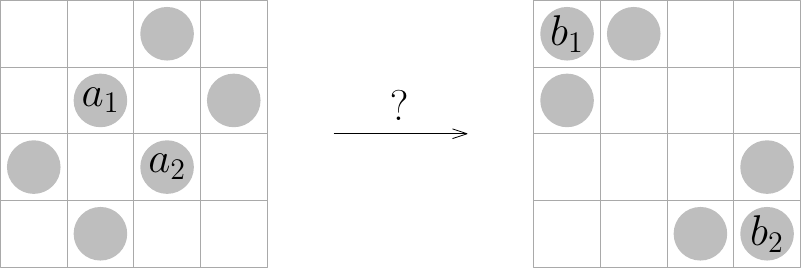}
	\caption{An example of a puzzle with two extra coins $a_1,a_2$ (Definition \ref{def_extra}) and two redundant coins $b_1,b_2$ (Definition \ref{def_redundant}).}\label{Example_square}
\end{figure}

\noindent For example, in the puzzle $A \xrightarrow{?} B$ from Figure \ref{Example_square}, $\{a_1,a_2\}$ is a set of extra coins in $A$ (in particular, it is also a set of extra coins in $A$ relatively to $B$).

\begin{proposition}
	Let $A$ and $B$ be distinct configurations. If $A \to B$ then $A$ has one extra coin relatively to $B$.
\end{proposition}

\begin{proof}
	Suppose $A$ does not have one extra coin relatively to $B$, and consider the first move $c \mapsto p$ made from $A$. This move can be decomposed as follows: first we remove $c$ from the board, then we put it back at $p$. After removing $c$ from the board, the span is $\Span(A \setminus \{c\})$. When we put the coin back, the span stays the same because of the 2-adjacency rule. Therefore the span after the first move is $\Span(A \setminus \{c\})$, which does not contain $\Span(B)$ because $\{c\}$ is not a set of extra coins relatively to $B$. By Proposition \ref{span_necessary}, the span cannot increase during the moves, hence $A \not\to B$.
\end{proof}

\noindent We now introduce the notion of redundant coins, which is also present in \cite{DDV02} (though unnamed). They are, to the target configuration $B$, the relevant analogue of what extra coins are to the starting configuration $A$.

\begin{definition}\label{def_redundant}
	Let $B$ be a configuration. A \textit{set of redundant coins in $B$} is a subset $B' \subset B$ of the form $B'=\{b_1,\ldots,b_k\}$ where, for all $i \in \{1,\ldots,k\}$, $b_i$ has at least two neighbors in $B \setminus \{b_1,\ldots,b_{i-1}\}$. We say $B$ \textit{has $k$ redundant coins} if it contains a set of $k$ redundant coins.
\end{definition}

\noindent For example, in the puzzle $A \xrightarrow{?} B$ from Figure \ref{Example_square}, $\{b_1,b_2\}$ is a set of redundant coins in $B$. Note that redundant coins are extra coins, but the converse is not true in general (think of $a_1$ and $a_2$ in Figure \ref{Example_square}).

\begin{proposition}\label{necessary_redundant}
	Let $A$ and $B$ be distinct configurations. If $A \to B$ then $B$ has one redundant coin.
\end{proposition}

\begin{proof}
	This is exactly Proposition \ref{trivial1}.
\end{proof}

\begin{proposition}\label{prop_redundant}
	Let $B$ be a configuration and let $B'=\{b_1,\ldots,b_k\}$ be a set of redundant coins in $B$, ordered as in the definition. Then $B \leftrightarrow (B \setminus B')^{+k}$. In particular, for any configuration $A$, we have $A \to B$ if and only if $A \to (B \setminus B')^{+k}$.
\end{proposition}

\begin{proof}
	We perform $B \to (B \setminus B')^{+k}$ by picking up the coins in $B'$, and $(B \setminus B')^{+k} \to B$ by dropping the coins in hand at $b_k,b_{k-1},\ldots,b_1$ successively (which respects the 2-adjacency rule since $b_i$ has at least two neighbors in $B \setminus \{b_1,\ldots,b_{i-1}\}$).
\end{proof}

\noindent Therefore, more redundant coins can only make a puzzle easier. One can think of extra coins in $A$ (or extra coins in $A$ relatively to $B$) as the first coins that we move, and of redundant coins in $B$ as the last coins that we place.

\vspace{1\baselineskip}
\section{Canonical configurations}\label{Section2}
\vspace{1\baselineskip}

In \cite{DDV02}, the authors define a reference minimum configuration for a given span, called the canonical configuration. Their method to solve a puzzle $A \xrightarrow{?} B$, if $\Span(A)=\Span(B)$ for instance, consists in going from $A$ to $B$ by routing through their common canonical configuration. We will reuse this principle in Sections \ref{Section3}, \ref{Section4} and \ref{Section5}. This section provides a summary of the definitions and results about canonical configurations that are presented in \cite{DDV02}.

\subsection{Definitions}

\begin{notation}
	We denote by $\dist$ the usual distance in the square grid.
\end{notation}

\begin{definition}{\textup{\cite{DDV02}}}
	A \textit{chain} between some coins $c$ and $c'$ is the configuration denoted by $[c_1,\ldots,c_N]$ which is formed by a sequence of coins $(c=c_1,c_2,\ldots,c_N=c')$ such that $\dist(c_i,c_{i+1}) \in \{1,2\}$ for all $i \in \{1,\ldots,N-1\}$.
\end{definition}

\begin{proposition}{\textup{\cite{DDV02}}}\label{prop_span_chain}
	The span of a chain coincides with its smallest enclosing rectangle.
\end{proposition}
	
\begin{definition}{\textup{\cite{DDV02}}}\label{Definition_L}
	An ‘L’ of size $m \times n$ is a minimum chain $L$ between two opposite corners of an $m \times n$ rectangle $R$ and hugging two consecutive sides of $R$. We say $L$ is \textit{even} (resp. \textit{odd}) if $R$ is even (resp. odd) i.e. if $m+n$ is even (resp. odd). See Figure \ref{canonical} (right).
\end{definition}
		
\noindent In accordance with Proposition \ref{prop_minimum}, an $m \times n$ ‘L’ has cardinality $\left\lceil\frac{m+n}{2}\right\rceil$, and consecutive coins in an even ‘L’ are at distance exactly 2 whereas consecutive coins in an odd ‘L’ are at distance exactly 2 except for a single pair of adjacent coins. An even ‘L’ is entirely defined by its span and orientation, whereas for an odd ‘L’ we also need the localization of the two adjacent coins. By Proposition \ref{prop_span_chain}, the span of an ‘L’ is its smallest enclosing rectangle i.e. $R$ in Definition \ref{Definition_L}.

\begin{definition}{\textup{\cite{DDV02}}}
	Let $R$ be an $m \times n$ rectangle: the \textit{canonical ‘L’ with span $R$} is the ‘L’ with span $R$ that is oriented like the letter L, with the additional property if $R$ is odd that the two adjacent coins are in the top-left corner (if $n$ is even) or bottom-right corner (if $m$ is even). Let $C$ be a configuration with span $\bigcup_{i=1}^s R_i$ where $R_1,\ldots,R_s$ are rectangles at distance at least 3 from each other, as per Proposition \ref{span_carre}: the \textit{canonical configuration associated to $C$} is the configuration denoted by $L_C$ with same span as $C$ such that, for all $1 \leq i \leq s$, $L_C \cap R_i$ is the canonical ‘L’ with span $R_i$. See Figure \ref{canonical}.
\end{definition}
		
\begin{figure}[h]
	\centering
	\includegraphics[scale=.4]{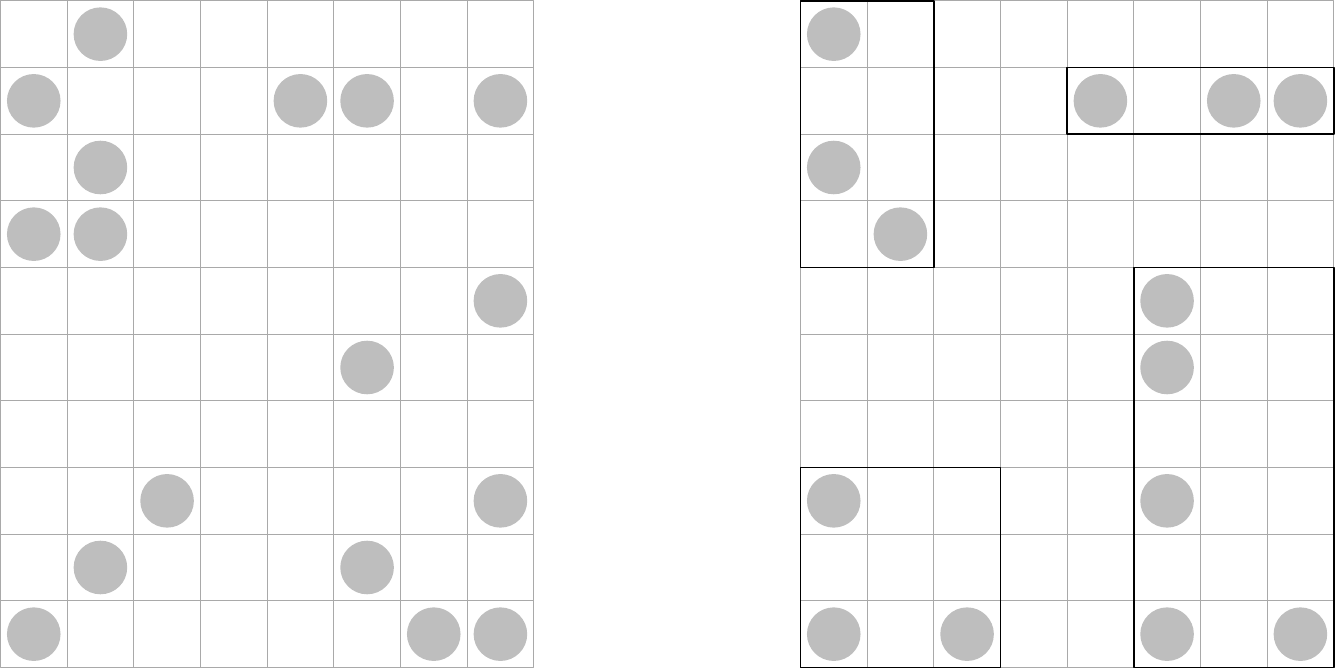}
	\caption{Left: a configuration $C$. Right: the associated canonical configuration $L_C$ (the top-left and bottom-left ‘L’s are even, the top-right and bottom-right ‘L’s are odd).}\label{canonical}
\end{figure}

\subsection{Transformations of ‘L’s}

\noindent Transformations of ‘L’s are the main subroutines used in the solving algorithms from \cite{DDV02}. We have the following:

\begin{proposition}{\textup{\cite{DDV02}}}
	Let $L_1$ and $L_2$ be two ‘L’s with same span. Then $L_1^{+2} \leftrightarrow L_2^{+2}$.
\end{proposition}

\noindent In particular, two coins in hand are enough to \textit{flip} any ‘L’, which means turning it into the mirrored ‘L’ hugging the other two sides of the span. Let us introduce our own routines for this particular transformation, one reason being that our method to flip odd ‘L’s only uses one extra coin, which will be useful in Section \ref{Section5}. Figure \ref{Flipping_even2} explains how to flip an even ‘L’ with two coins in hand, using the subroutines from Figure \ref{Flipping_even}. Figure \ref{Flipping_odd2} explains how to flip an even ‘L’ with one coin in hand, using the subroutines from Figure \ref{Flipping_odd} as well as the \textit{leapfrog} technique described in Figure \ref{Leapfrog}: a leapfrog means relocating the unique pair of adjacent coins inside of an odd ‘L’. In all figures throughout this paper, an encircled coin represents a coin that has just been dropped, while a crossed out coin represents a coin that we pick up.
\begin{figure}[!htb]
	\centering
	\includegraphics[scale=.4]{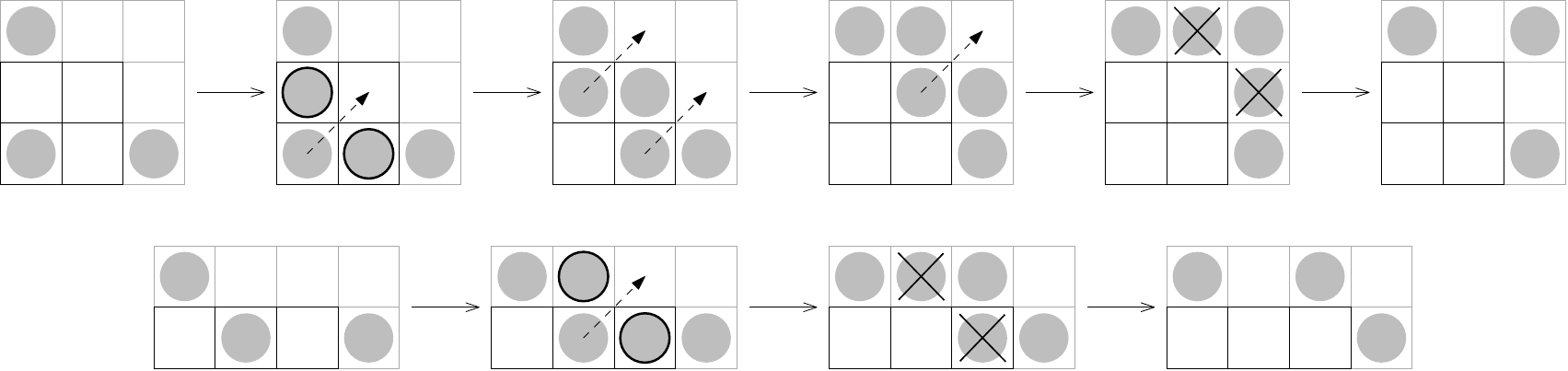}
	\caption{Subroutines used to flip an even ‘L’. The bottom subroutine is only used if both sides are even.}\label{Flipping_even}
\end{figure}
\begin{figure}[!htb]
	\centering
	\includegraphics[scale=.345]{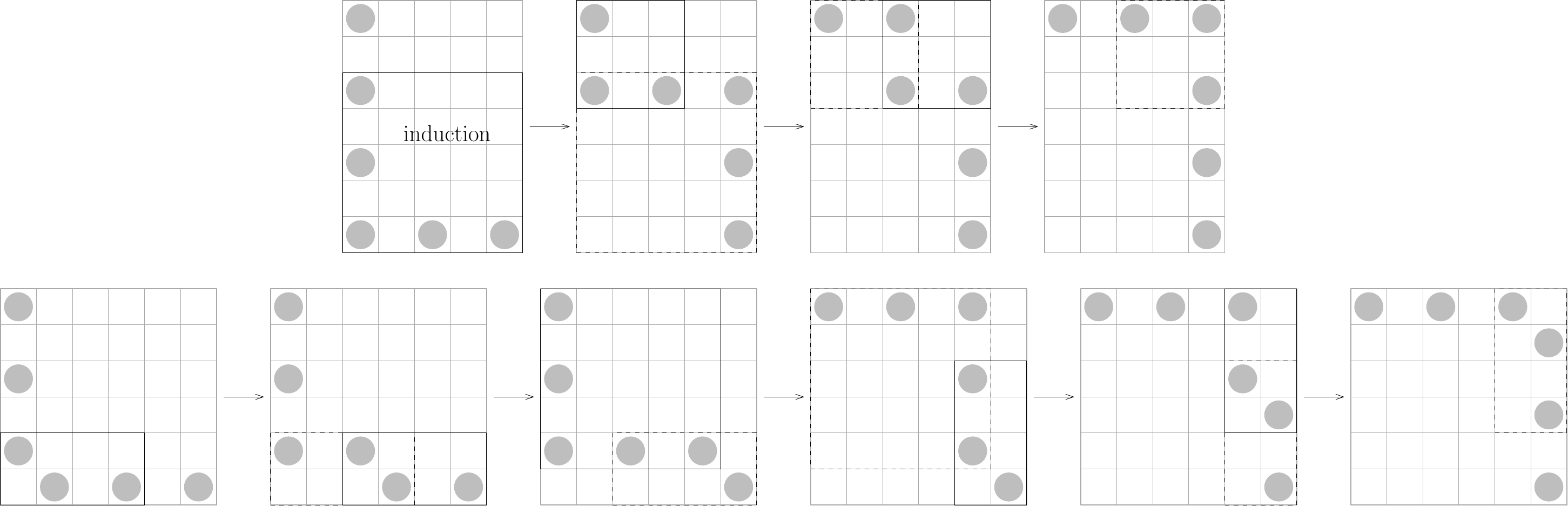}
	\caption{Flipping an even ‘L’, with odd sides (top) or even sides (bottom).}\label{Flipping_even2}
\end{figure}
\begin{figure}[!htb]
	\centering
	\includegraphics[scale=.393]{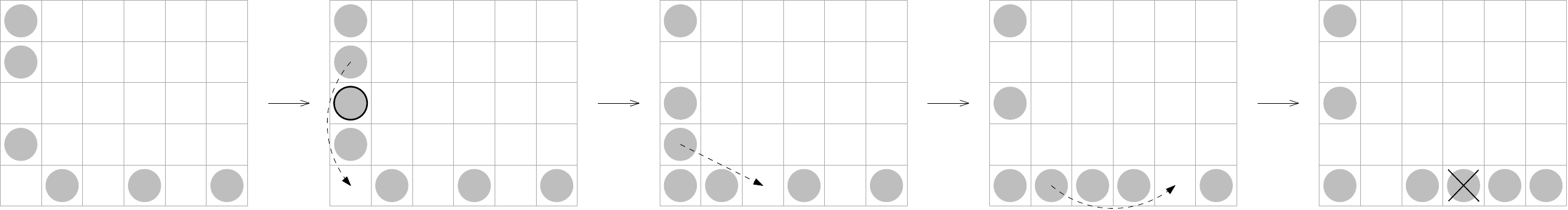}
	\caption{A leapfrog. Intermediary states cover all possible locations of the pair of adjacent coins.}\label{Leapfrog}
\end{figure}
\begin{figure}[!htb]
	\centering
	\includegraphics[scale=.383]{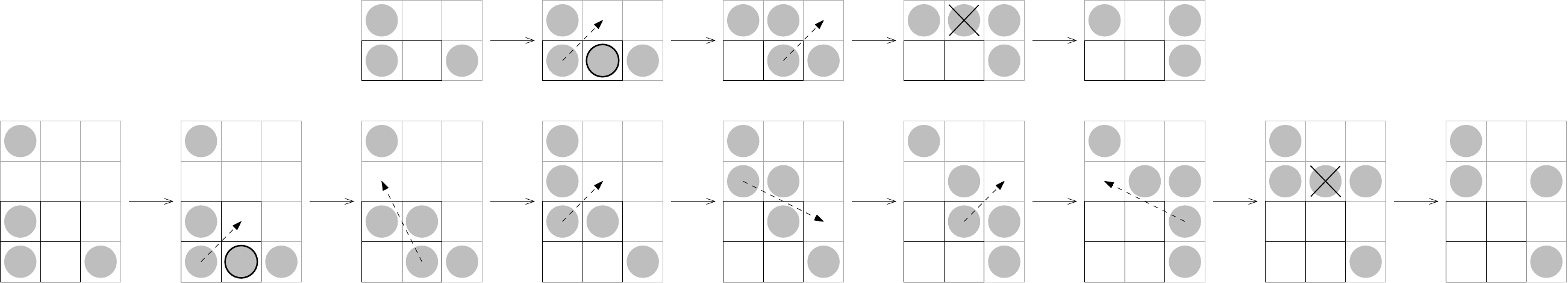}
	\caption{Subroutines used to flip an odd ‘L’.}\label{Flipping_odd}
\end{figure}
\begin{figure}[!htb]
	\centering
	\includegraphics[scale=.4]{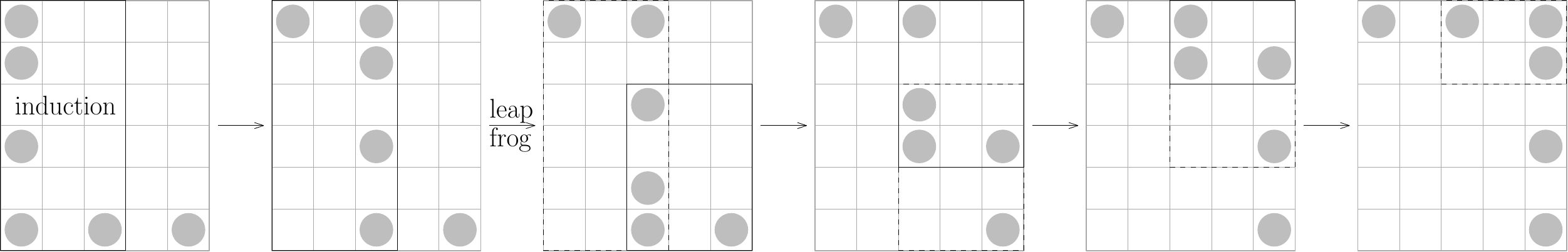}
	\caption{Flipping an odd ‘L’.}\label{Flipping_odd2}
\end{figure}

\subsection{Canonicalization process}

The crucial result is that two coins in hand are enough to turn any configuration into its associated canonical configuration in a reversible manner:

\begin{lemme}\label{main_lemma}{\textup{\cite{DDV02}}}
	For any configuration $C$, we have $C^{+2} \leftrightarrow L_C^{+2+|C|-|L_C|}$.
\end{lemme}

\noindent From there, a method to solve a puzzle $A \xrightarrow{?} B$ would roughly be to: pick up two coins in $A$; canonicalize; reverse into $B$ minus two coins; drop two coins to finish $B$. To do this however:
\begin{itemize}[noitemsep,nolistsep]
	\item[--] We need two coins in $A$ that we can pick up without breaking the inclusion of spans at the start i.e. 2 extra coins in $A$ relatively to $B$.
	\item[--] We need two appropriate spots in $B$ to drop our two coins in hand at the end i.e. 2 redundant coins in $B$.
\end{itemize}
Moreover, unless the spans are equal once the first two coins $\{a_1,a_2\}$ have been picked up, we need a way to go from $L_{A \setminus \{a_1,a_2\}}$ to $L_B$. This reasoning can be summed up as follows:

\begin{corollaire}\label{coro_main_lemma}
	Let $A$ and $B$ be configurations such that $|A|=|B|$ and:
	 \begin{enumerate}[noitemsep,nolistsep,label=(\roman*)]
		\item $A$ has 2 extra coins relatively to $B$.
		\item $B$ has 2 redundant coins.
	\end{enumerate}
	Let $A_0 \defeq A \setminus \{a_1,a_2\}$ where $\{a_1,a_2\}$ is a set of extra coins in $A$ relatively to $B$. \\ If $L_{A_0}^{+|A|-|L_{A_0}|} \to L_B^{+|B|-|L_B|}$, then $A \to B$.
\end{corollaire}

\begin{proof}
	Let $B_0 \defeq B \setminus \{b_1,b_2\}$ where $\{b_1,b_2\}$ is a set of redundant coins in $B$. Note that $L_B=L_{B_0}$ since $\Span(B)=\Span(B_0)$. We go from $A$ to $B$ in five steps:
	\begin{itemize}
		\item We get $A \to A_0^{+2}$ by picking up $a_1$ and $a_2$.
		\item Lemma \ref{main_lemma} ensures that $A_0^{+2} \to L_{A_0}^{+2+|A_0|-|L_{A_0}|}$.
		\item We have $L_{A_0}^{+2+|A_0|-|L_{A_0}|}=L_{A_0}^{+|A|-|L_{A_0}|} \to L_B^{+|B|-|L_B|}=L_{B_0}^{+2+|B_0|-|L_{B_0}|}$ by assumption.
		\item Lemma \ref{main_lemma} ensures that $L_{B_0}^{+2+|B_0|-|L_{B_0}|} \to B_0^{+2}$.
		\item Finally, Proposition \ref{prop_redundant} yields $B_0^{+2} \to B$. \qedhere
	\end{itemize}
	\renewcommand{\qedsymbol}{}
\end{proof}

\vspace{1\baselineskip}
\section{Two extra coins: with additional constraints}\label{Section3}
\vspace{1\baselineskip}

In \cite{DDV02}, the authors studied the case where $A$ has 2 extra coins relatively to $B$, which seemed enough to solve most puzzles. In this section, we go back on their result and then give it a slight improvement.

\subsection{Previous result}\label{previous}
	    
The main result in \cite{DDV02} is the following:    

\begin{theoreme}\label{theoreme_plusdeux}{\textup{\cite{DDV02}}}
	Let $A$ and $B$ be configurations such that $|A|=|B|$, and suppose that:
	\begin{enumerate}[noitemsep,nolistsep,label=(\roman*)]
		\item $\Span(A)=\Span(B)$.
		\item $A$ has 2 extra coins relatively to $B$.
		\item $B$ has 2 redundant coins.
	\end{enumerate}
	Then $A \to B$.
\end{theoreme}

\noindent As mentioned in Section \ref{Section2}, this theorem uses canonical configurations as an intermediary. Indeed, it follows immediately from Corollary \ref{coro_main_lemma} (with $L_{A_0}=L_A=L_B$ since $\Span(A)=\Span(B)$). An example of a puzzle that is solvable according to Theorem \ref{theoreme_plusdeux} is the one in Figure \ref{Example_square}, where $\{a_1,a_2\}$ is a set of extra coins in $A$ relatively to $B$ and $\{b_1,b_2\}$ is a set of redundant coins in $B$.

\vspace{1\baselineskip}
\noindent In fact, a stronger version of this theorem is claimed in \cite{DDV02}, where condition \textit{(i)} is not required. The authors reduce to the case where $\Span(A)=\Span(B)$ by picking up all coins in $A \setminus \Span(B)$. However, this does not work in general, because this might actually cause the span to become strictly smaller than that of $B$. It turns out that, without the added condition \textit{(i)}, some puzzles are solvable and some are not, as illustrated by Figure \ref{problem}. The puzzle on the left is solvable in 12 moves. The puzzle on the right is unsolvable (as we will later prove) and therefore is a counterexample to the version of the theorem in \cite{DDV02}. It is easy to check that no smaller counterexample exists, be it in terms of number of coins or half-perimeter of the starting span. A generalized family of counterexamples will be exhibited in Section \ref{Section4}: in all of them, the problem comes from the fact that $\Span(B)$ is split into two far apart components.

\begin{figure}[h]
	\centering
	\includegraphics[scale=.4]{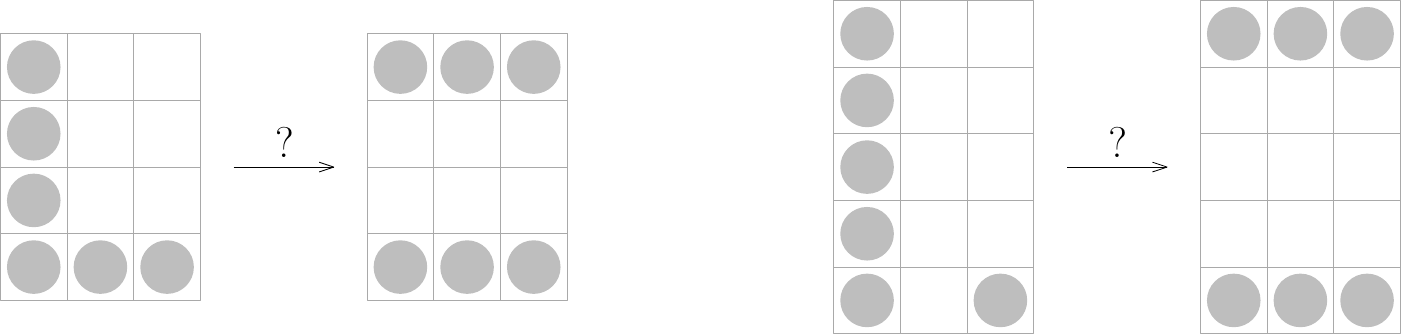}
	\caption{Two puzzles satisfying conditions \textit{(ii)} and \textit{(iii)} but not \textit{(i)}, since $\Span(A) \supsetneq \Span(B)$. The left one is solvable but the right one is not.}\label{problem}
\end{figure}

\subsection{A slight improvement}

First of all, we show that Theorem \ref{theoreme_plusdeux} can be slightly improved: instead of asking for the spans to be equal, it is sufficient that each component of the starting span contains at most one component of the target span. Indeed, while splitting an ‘L’ into two separate components can be difficult (as we have just seen), shrinking an ‘L’ with two coins in hand is not a problem.

\begin{lemme}\label{lemma_shrinking}
	If $L_1$ and $L_2$ are canonical ‘L’s with $\Span(L_1) \supseteq \Span(L_2)$, then $L_1^{+2} \to L_2^{+2+|L_1|-|L_2|}$.
\end{lemme}

\begin{proof}
	First of all, we \textit{trim} $L_1$ to the right if needed, as follows (see Figure \ref{Shrinking}):
	\begin{enumerate}[noitemsep,nolistsep,label={\arabic*.}]
		\item If $L_1$ is odd, we use a leapfrog to put the pair of adjacent coins to the far right.
		\item We make sure there is a coin $c$ at the rightmost position that we want to keep, by dropping one there if needed.
		\item We finish by simply picking up all coins that are further right than $c$.
	\end{enumerate}
	\begin{figure}[h]
		\centering
		\includegraphics[scale=.4]{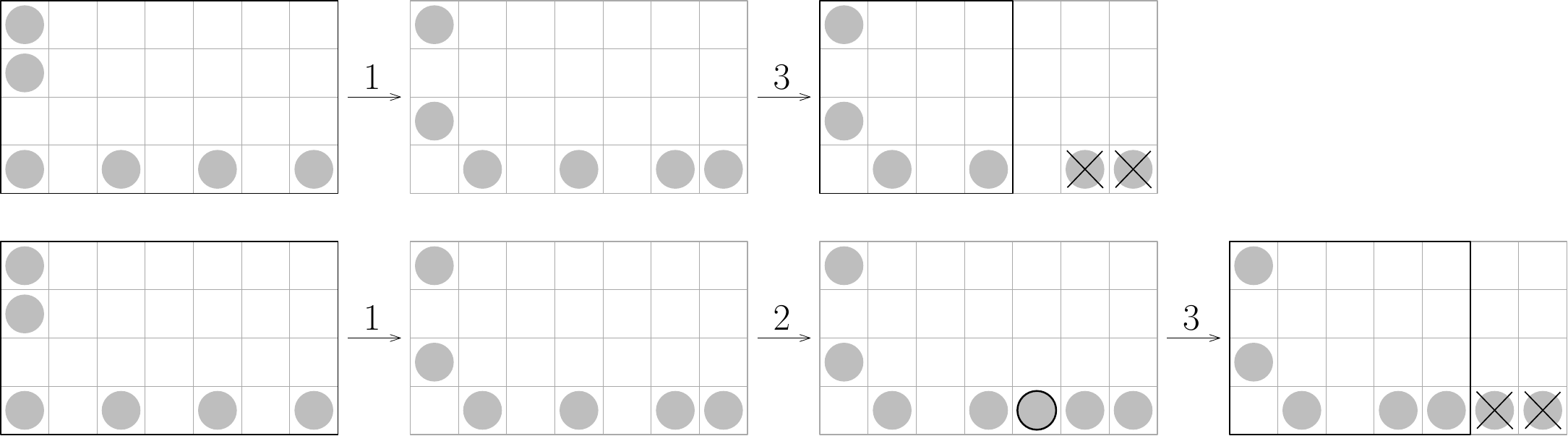}
		\caption{Trimming a $7 \times 4$ ‘L’ to its right, making it $4 \times 4$ (top) or $5 \times 4$ (bottom). The numbers above the arrows refer to the three steps.}\label{Shrinking}
	\end{figure}
	We then trim our ‘L’ at the top, in analogous fashion. We now flip it, so it is now ready to be trimmed to the left and at the bottom. Once this is done, we flip it back and use a leapfrog if needed to make it canonical.
\end{proof}

\begin{theoreme}\label{theoreme_plusdeux2}
	Let $A$ and $B$ be configurations such that $|A|=|B|$, and suppose that:
	\begin{enumerate}[noitemsep,nolistsep,label=(\roman*)]
		\item $A$ has 2 extra coins relatively to $B$, and more precisely: there exist $a_1 \neq a_2$ in $A$ such that $\Span(A \setminus \{a_1,a_2\}) \supseteq \Span(B)$ \underline{and} each component of $\Span(A \setminus \{a_1,a_2\})$ contains at most one component of $\Span(B)$.
		\item $B$ has 2 redundant coins.
	\end{enumerate}
	Then $A \to B$.
\end{theoreme}

\begin{proof}
	Let $A_0 \defeq A \setminus \{a_1,a_2\}$. By Corollary \ref{coro_main_lemma}, it suffices to show that $L_{A_0}^{+|A|-|L_{A_0}|} \to L_B^{+|B|-|L_B|}$. Since each component of $\Span(A_0)$ contains at most one component of $\Span(B)$, we can use Lemma \ref{lemma_shrinking} to shrink each ‘L’ in $L_{A_0}$ to the size of the corresponding ‘L’ in $L_{B_0}$. This is always possible, because we start off with $|A|-|L_{A_0}|\geq |A|-|A_0|=2$ coins in hand and this number cannot decrease each time we shrink an ‘L’.
\end{proof}

\noindent Note that puzzles satisfying condition \textit{(i)} but not condition \textit{(ii)} (while still satisfying the fact that $B$ has 1 redundant coin, otherwise we would be in a trivially unsolvable case by Proposition \ref{necessary_redundant}) also may or may not be solvable as shown in Figure \ref{problem2}.
	
\begin{figure}[h]
	\centering
	\includegraphics[scale=.4]{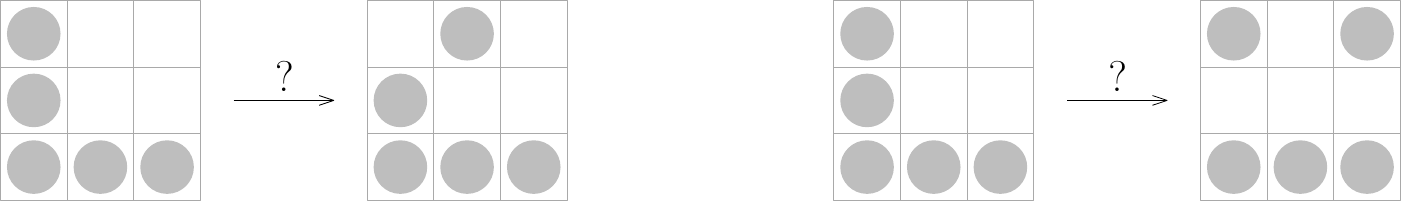}
	\caption{Two puzzles satisfying condition \textit{(i)} but not \textit{(ii)}. The left puzzle is solvable in 4 moves, while the right puzzle is unsolvable.}\label{problem2}
\end{figure}

\vspace{1\baselineskip}
\section{Two extra coins: general case}\label{Section4}
\vspace{1\baselineskip}

What if some component of $\Span(A)$ contains two or more components of $\Span(B)$, so that Theorem \ref{theoreme_plusdeux2} does not apply? A first natural guess would be that we then need more extra coins and/or more redundant coins than just two. Nevertheless, we now exhibit a family of unsolvable puzzles which proves that, even with the inclusion of spans, no constant number of extra coins in $A$ (relatively to $B$ or not) or redundant coins in $B$ can guarantee that a puzzle is solvable in general. Next, we present a new sufficient condition, which shows in particular that the aforementioned family consists of just about worst-case puzzles.

\subsection{Worst-case puzzles}\label{Section4-1}

Puzzles like those from Figure \ref{problem} require to split the span, at some point during the moves, in a way that we now prove impossible without a certain amount of total coins relative to the size of the rectangles involved.

\begin{definition}\label{def_split}
	Let $R_1$ and $R_2$ be rectangles at distance at least 3 from each other, and let $A$ be a configuration such that $R_1$ and $R_2$ are included in the same component of $\Span(A)$. An \textit{$(R_1,R_2)$-split of $A$} is a sequence of moves $A=A_0 \mapsto A_1 \mapsto \ldots \mapsto A_T$ ($T \geq 1$ necessarily) such that $R_1$ and $R_2$ are included in two separate components of $\Span(A_T)$.
\end{definition}

\begin{proposition}\label{prop_counterexample}
	Let $R_1$ and $R_2$ be rectangles of size $m_1 \times n_1$ and $m_2 \times n_2$ respectively, whose projections on the $x$ axis intersect, and whose projections on the $y$ axis do not intersect with a gap of $h \geq 2$ rows separating them. Let $A$ be a configuration such that $R_1$ and $R_2$ are included in the same component of $\Span(A)$. If there exists an $(R_1,R_2)$-split of $A$, then $|A| \geq \frac{m_1+n_1+m_2+n_2+h-1}{2}$.
\end{proposition}

\begin{proof}
	Let $A=A_0 \mapsto A_1 \mapsto \ldots \mapsto A_T$ be an $(R_1,R_2)$-split of $A$ with minimum number of moves, so that $R_1$ and $R_2$ are included in the same component $R$ of $\Span(A_{T-1})$ but in two separate components $R'_1$ and $R'_2$ of $\Span(A_T)$ (see Figure \ref{counterexample_ajout}). If $R'_1$ is of size $m'_1 \times n'_1$ and $R'_2$ is of size $m'_2 \times n'_2$ with a gap of $h'$ rows separating them, then we have $m'_1 \geq m_1$, $m'_2 \geq m_2$ and $n'_1+n'_2+h' \geq n_1+n_2+h$ hence $\frac{m'_1+n'_1+m'_2+n'_2+h'-1}{2} \geq \frac{m_1+n_1+m_2+n_2+h-1}{2}$. Therefore, the worst case for what we want to prove is if $R_1=R'_1$ and $R_2=R'_2$, which is what we assume from now.

	\begin{figure}[h]
		\centering
		\includegraphics[scale=.45]{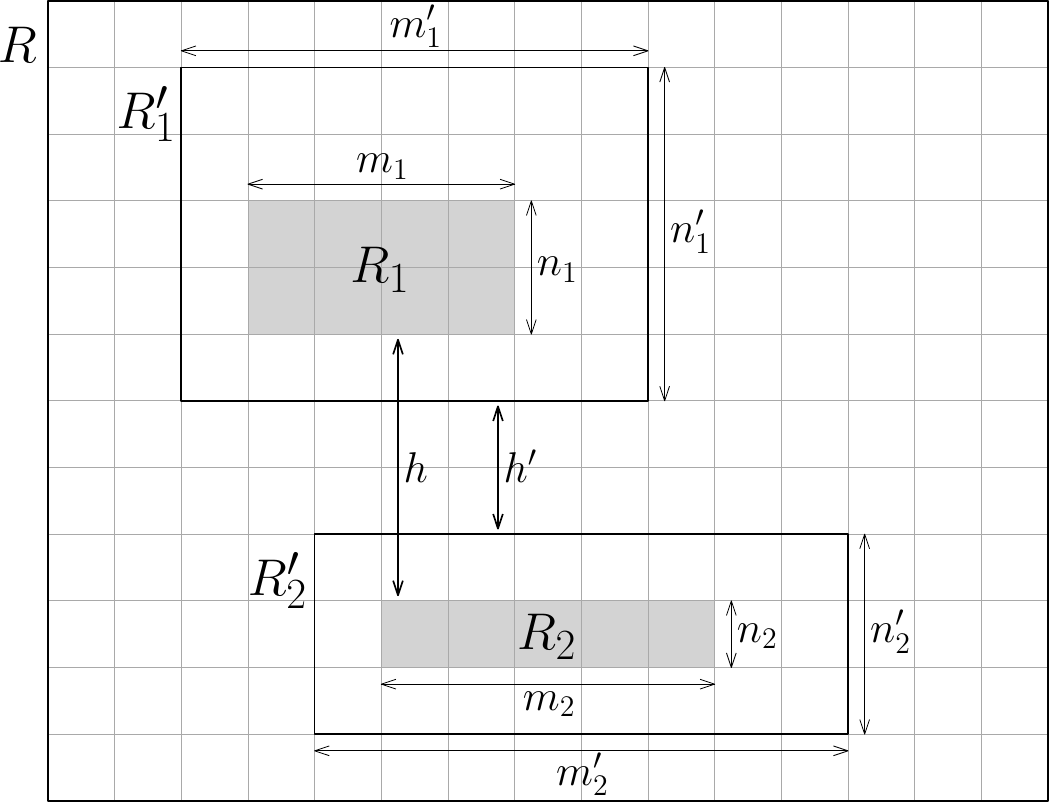}
		\caption{Illustration of $R$, $R_1$, $R_2$, $R'_1$, $R'_2$.}\label{counterexample_ajout}
	\end{figure}
	
	\noindent We use $A_{T-1}$ to count the coins and carry out the proof. We have:
	\begin{equation}\label{cardinality}
		|A|=|A_{T-1}|= |A_{T-1} \cap R_1| + |A_{T-1} \cap R_2| +|A_{T-1} \setminus (R_1 \cup R_2)|.
	\end{equation}
	
	\noindent Moreover:
	\begin{itemize}[noitemsep,nolistsep]
		\item The only coin in $A_T \cap R_1$ that might not be in $A_{T-1} \cap R_1$ is the coin that has been moved to go from $A_{T-1}$ to $A_T$, however that coin does not contribute to the span since it has at least two coins adjacent to it in $A_T \cap R_1$. Since $\Span(A_T \cap R_1)=R_1$, we thus get $\Span(A_{T-1} \cap R_1)=R_1$. In particular, Proposition \ref{prop_minimum} yields:
		\begin{equation}\label{cardinality1}
			|A_{T-1} \cap R_1| \geq \left\lceil\frac{m_1+n_1}{2}\right\rceil.
		\end{equation}
		\item Similarly, $\Span(A_{T-1} \cap R_2)=R_2$ and:
		\begin{equation}\label{cardinality2}
			|A_{T-1} \cap R_2| \geq \left\lceil\frac{m_2+n_2}{2}\right\rceil.
		\end{equation}
		\item By Proposition \ref{span_carre2}, there cannot be two consecutive rows of $R$ without a coin in $A_{T-1}$. Since there is a gap of $h$ rows between $R_1$ and $R_2$, this yields:
		\begin{equation}\label{cardinality3}
			|A_{T-1} \setminus (R_1 \cup R_2)| \geq \left\lfloor \frac{h}{2} \right\rfloor.
		\end{equation}
	\end{itemize}
	\noindent Combining (\ref{cardinality}), (\ref{cardinality1}), (\ref{cardinality2}) and (\ref{cardinality3}), we get:
	\[ |A| \geq \left\lceil\frac{m_1+n_1}{2}\right\rceil + \left\lceil\frac{m_2+n_2}{2}\right\rceil + \left\lfloor \frac{h}{2} \right\rfloor \geq \frac{m_1+n_1+m_2+n_2+h-1}{2}. \qedhere \]
	
\renewcommand{\qedsymbol}{}

\end{proof}

\begin{figure}[h]
	\centering
	\includegraphics[scale=.4]{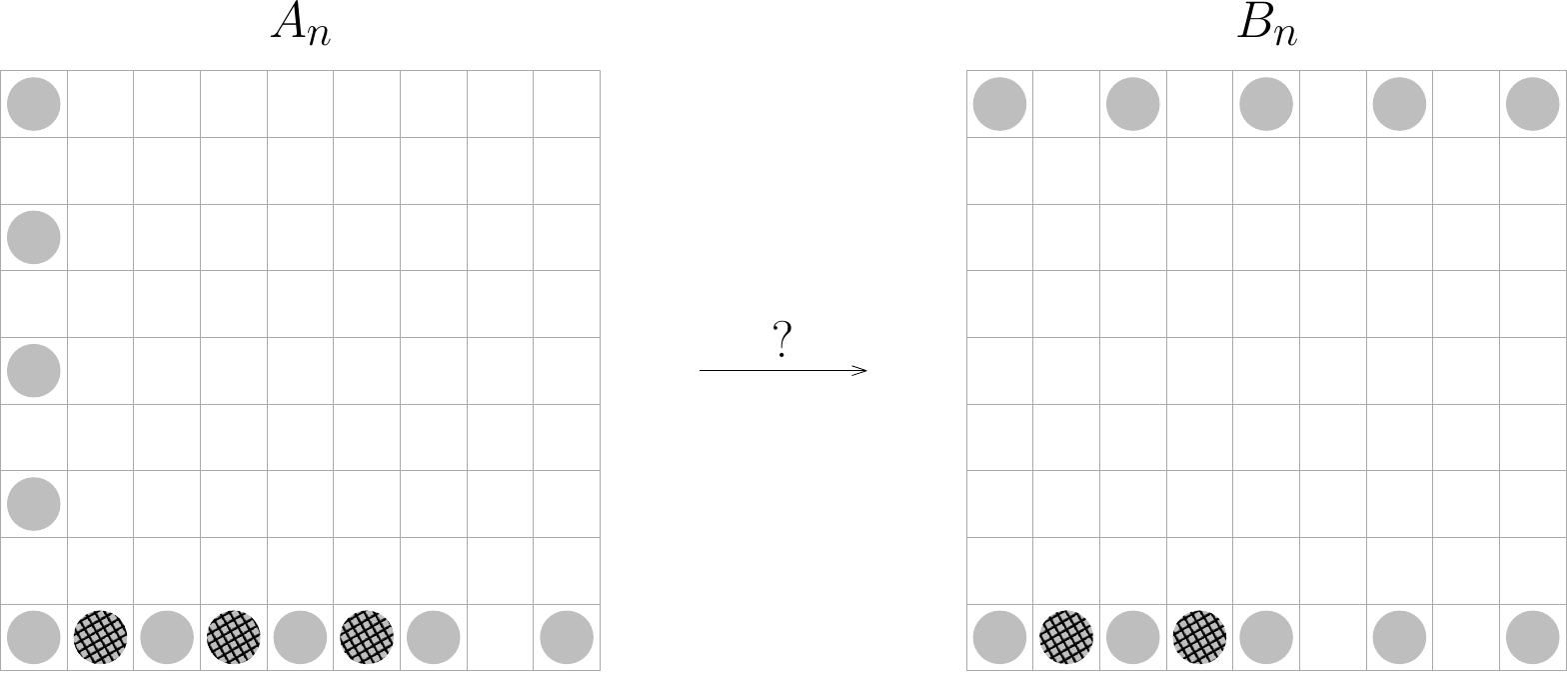}
	\caption{Definition of the puzzle $A_n \xrightarrow{?} B_n$ (here $n=9$). The shaded coins represent the extra/redundant coins.}\label{counterexample2}
\end{figure}

\begin{corollaire}\label{coro_counterexample}
	Let $n \geq 9$. We define configurations $A_n$ and $B_n$ as in Figure \ref{counterexample2}:
	\begin{itemize}[noitemsep,nolistsep]
		\item $A_n$ consists of an $n \times n$ ‘L’ with $\left\lfloor\frac{n-2}{2}\right\rfloor$ coins added to the bottom row.
		\item $B_n$ has the same smallest enclosing rectangle as $A_n$ and contains an $n \times 1$ ‘L’ in both the top and bottom row with $\left\lfloor\frac{n-5}{2}\right\rfloor$ coins added to the bottom row.
	\end{itemize}
	Then $A_n \not\to B_n$ even though: $A_n$ has $\left\lfloor\frac{n-2}{2}\right\rfloor$ extra coins, $B_n$ has $\left\lfloor\frac{n-5}{2}\right\rfloor$ redundant coins, and $\Span(A_n) \supset \Span(B_n)$.
\end{corollaire}

\begin{proof}
	Let $R_1$ (resp. $R_2$) be the top row (resp. the bottom row) of $\Span(A_n)$: $R_1$ and $R_2$ are of size $n \times 1$ and separated by a gap of $n-2$ rows. Solving this puzzle would mean performing an $(R_1,R_2)$-split of $A_n$, which is impossible by Proposition \ref{prop_counterexample} because:
	\begin{itemize}[noitemsep,nolistsep]
		\item If $n$ is odd then $n+\frac{n-3}{2}=|A_n|=|B_n|=2\frac{n+1}{2}+\frac{n-5}{2}=\frac{3n-3}{2}<\frac{n+1+n+1+(n-2)-1}{2}$.
		\item If $n$ is even then $n+\frac{n-2}{2}=|A_n|=|B_n|=2\frac{n+2}{2}+\frac{n-6}{2}=\frac{3n-2}{2}<\frac{n+1+n+1+(n-2)-1}{2}$. \qedhere
	\end{itemize}
	\renewcommand{\qedsymbol}{}
\end{proof}

\begin{corollaire}\label{coro_counterexample2}
    For any $k \in \N$, there exist configurations $A$ and $B$ with $|A|=|B|$ such that:
    \begin{itemize}[noitemsep,nolistsep]
		\item $\Span(A) \supset \Span(B)$.
		\item $A$ has $k$ extra coins.
		\item $B$ has $k$ redundant coins.
		\item $A \not\to B$.
	\end{itemize}
\end{corollaire}

\noindent It is actually possible to improve Proposition \ref{prop_counterexample}: we now give a refined bound that even applies to the puzzle on the right of Figure \ref{problem} which, as previously mentioned, is the smallest counterexample to the version of Theorem \ref{theoreme_plusdeux} in \cite{DDV02}.

\begin{proposition}\label{prop_counterexample2}
    Let $R_1$, $R_2$ and $A$ be as in Proposition \ref{prop_counterexample}. If there exists an $(R_1,R_2)$-split of $A$ but none in two moves or less, then $|A| \geq \frac{m_1+n_1+m_2+n_2+h+2}{2}$.
\end{proposition}

\begin{proof}
    Let us pick up where the proof of Proposition \ref{prop_counterexample} ended. Since there exists no $(R_1,R_2)$-split of $A$ in two moves or less, we have $T \geq 3$. We use the fact that $A_{T-1}$ then has the following properties:
	\begin{enumerate}[noitemsep,nolistsep,label=(\roman*)]
		\item $A_{T-1}$ contains a coin that is adjacent to at least two other coins.
		\item For all $c \in A_{T-1}$, $A_{T-1} \setminus \{c\}$ does not consist of all isolated coins.
	\end{enumerate}
	Property (i) comes from Proposition \ref{trivial1} since $T-1 \geq 1$, and property (ii) comes from Proposition \ref{trivial2} since $T-1 \geq 2$ (indeed $A \not\mapsto A_{T-1}$ because our sequence of moves has been chosen shortest). We distinguish between four cases:
	
	\begin{enumerate}[label=\arabic*)]
	
		\item Case 1: $R_1$ and $R_2$ are both odd; $h$ is even.
		\\ This is the easiest case:
		\begin{align*} |A| \geq \left\lceil\frac{m_1+n_1}{2}\right\rceil + \left\lceil\frac{m_2+n_2}{2}\right\rceil + \left\lfloor \frac{h}{2} \right\rfloor & = \frac{m_1+n_1+1}{2} + \frac{m_2+n_2+1}{2} + \frac{h}{2} \\ & = \frac{m_1+n_1+m_2+n_2+h+2}{2}. \end{align*}
	
		\item Case 2: At least one of $R_1$ or $R_2$ is even (say $R_2$ is even); $h$ is even.
		\\ We just need to find one coin more than what (\ref{cardinality1}), (\ref{cardinality2}) and (\ref{cardinality3}) give us combined, because if we do then we can conclude that:
		\begin{align*} |A| & \geq \left\lceil\frac{m_1+n_1}{2}\right\rceil + \left\lceil\frac{m_2+n_2}{2}\right\rceil + \left\lfloor \frac{h}{2} \right\rfloor + 1 \\ & \geq \frac{m_1+n_1}{2} + \frac{m_2+n_2}{2} + \frac{h}{2} + 1 = \frac{m_1+n_1+m_2+n_2+h+2}{2}. \end{align*}
		Therefore, suppose for a contradiction that (\ref{cardinality1}), (\ref{cardinality2}) and (\ref{cardinality3}) are all tight. For (\ref{cardinality3}), this means that $A_{T-1} \setminus (R_1 \cup R_2)$ consists exactly of one coin every two rows in the gap between $R_1$ and $R_2$. For (\ref{cardinality1}) and (\ref{cardinality2}), this means $A_{T-1} \cap R_1$ and $A_{T-1} \cap R_2$ are both minimum: by Proposition \ref{prop_minimum}, all coins in $A_{T-1} \cap R_2$ are isolated and all coins in $A_{T-1} \cap R_1$ are isolated except possibly for a single pair of adjacent coins. We can see that the only way to satisfy property (i) is if $A_{T-1} \cap R_1$ contains a pair $\{c_1,c_2\}$ of adjacent coins such that $c_1$ is adjacent to one of the coins in $A_{T-1} \setminus (R_1 \cup R_2)$, as in Figure \ref{counterexample_ajout2} (left). However, $c_1$ then violates property (ii), a contradiction.
		
		\item Case 3: At least one of $R_1$ or $R_2$ is odd (say $R_2$ is odd); $h$ is odd.
		\\ Again, we just need to find one coin more than what (\ref{cardinality1}), (\ref{cardinality2}) and (\ref{cardinality3}) give us combined, because if we do then we can conclude that:
		\begin{align*} |A| & \geq \left\lceil\frac{m_1+n_1}{2}\right\rceil + \left\lceil\frac{m_2+n_2}{2}\right\rceil + \left\lfloor \frac{h}{2} \right\rfloor + 1 \\ & \geq \frac{m_1+n_1}{2} + \frac{m_2+n_2+1}{2} + \frac{h-1}{2} + 1 = \frac{m_1+n_1+m_2+n_2+h+2}{2}. \end{align*}
		Therefore, suppose for a contradiction that (\ref{cardinality1}), (\ref{cardinality2}) and (\ref{cardinality3}) are all tight. For (\ref{cardinality1}) and (\ref{cardinality2}), this means that $A_{T-1} \cap R_1$ and $A_{T-1} \cap R_2$ are both minimum: in particular, neither contains three coins such that one is adjacent to the other two. For (\ref{cardinality3}), this means that $A_{T-1} \setminus (R_1 \cup R_2)$ consists exactly of one coin every two rows in the gap between $R_1$ and $R_2$, none of which is adjacent to $R_1$ or $R_2$ since $h$ is odd. See Figure \ref{counterexample_ajout2} (middle). All in all, property (i) is violated, a contradiction.
		
		\item Case 4: $R_1$ and $R_2$ are both even; $h$ is odd.
		\\ In this case, we need to find two coins more than what (\ref{cardinality1}), (\ref{cardinality2}) and (\ref{cardinality3}) give us combined, because if we do then we can conclude that:
		\begin{align*} |A| & \geq \left\lceil\frac{m_1+n_1}{2}\right\rceil + \left\lceil\frac{m_2+n_2}{2}\right\rceil + \left\lfloor \frac{h}{2} \right\rfloor + 2 \\ & \geq \frac{m_1+n_1}{2} + \frac{m_2+n_2}{2} + \frac{h-1}{2} + 2 = \frac{m_1+n_1+m_2+n_2+h+3}{2}. \end{align*}
		First of all, the same proof as in Case 3 shows that (\ref{cardinality1}), (\ref{cardinality2}) and (\ref{cardinality3}) cannot all be tight: there needs to be some coin $c$ in $A_{t-1}$ that has (at least) two coins adjacent to it. Now suppose for a contradiction that all three inequalities become tight if we remove $c$ i.e.: $|(A_{T-1}\setminus \{c\}) \cap R_1| = \left\lceil\frac{m_1+n_1}{2}\right\rceil$, $|(A_{T-1}\setminus \{c\}) \cap R_2| = \left\lceil\frac{m_2+n_2}{2}\right\rceil$, $|(A_{T-1}\setminus \{c\}) \setminus (R_1 \cup R_2)| = \left\lfloor \frac{h}{2} \right\rfloor$. See Figure \ref{counterexample_ajout2} (right). If $c \in R_1$ then the two coins adjacent to $c$ in $A_{T-1}$ are inside $R_1$ as well (indeed, as we have seen in Case 3, the fact that $h$ is odd means that none of the $\lfloor\frac{h}{2}\rfloor$ coins in $A_{T-1} \setminus (R_1 \cup R_2)$ is adjacent to $R_1$). Therefore $(A_{T-1}\setminus \{c\}) \cap R_1$ has span $R_1$, and is minimum since $|(A_{T-1}\setminus \{c\}) \cap R_1| = \left\lceil\frac{m_1+n_1}{2}\right\rceil$. Similarly, $(A_{T-1}\setminus \{c\}) \cap R_2$ has span $R_2$ and is minimum. By Proposition \ref{prop_minimum}, all coins in $(A_{T-1} \setminus \{c\}) \cap R_1$ and $(A_{T-1} \setminus \{c\}) \cap R_2$ are thus isolated. Since $h$ is odd, the $\lfloor\frac{h}{2}\rfloor$ coins in $(A_{T-1}\setminus \{c\}) \setminus (R_1 \cup R_2)$ are also isolated. This contradicts property (ii). \qedhere
		
	\end{enumerate}
	
	\begin{figure}[h]
		\centering
		\includegraphics[scale=.4]{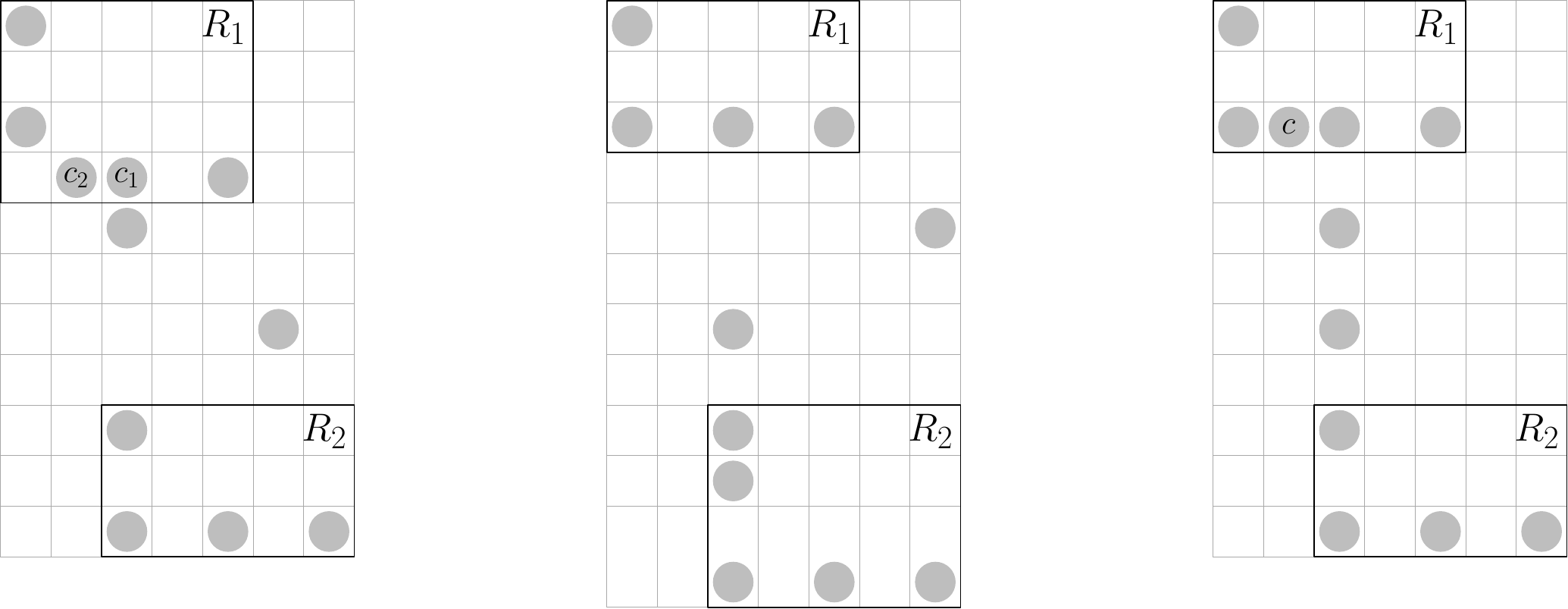}
		\caption{The configuration $A_{T-1}$: a contradiction in Case 2 (left), Case 3 (middle) and Case 4 (right).}\label{counterexample_ajout2}
	\end{figure}

	\renewcommand{\qedsymbol}{}

\end{proof}

\begin{corollaire}
	The puzzle on the right of Figure \ref{problem} is unsolvable.
\end{corollaire}

\begin{proof}
	Let $R_1$ (resp. $R_2$) be the top row (resp. the bottom row) of the starting span: $R_1$ and $R_2$ are of size $3 \times 1$ and separated by a three row gap. Solving this puzzle would mean performing an $(R_1,R_2)$-split of the starting configuration, which is impossible by Proposition \ref{prop_counterexample2}: indeed, we can easily check that it is impossible in two moves or less, and the puzzle contains $6< \frac{3+1+3+1+3+2}{2}$ coins.
\end{proof}

\subsection{A new sufficient condition}

We now present a result that holds even when some component of $\Span(A)$ contains two or more components of $\Span(B)$.

\begin{notation}
	Let $C$ be a configuration. We denote by $\min_C$ the cardinality of minimum configurations with same span as $C$. Note that, if this span is an $m \times n$ rectangle, then $\min_C=\left\lceil\frac{m+n}{2}\right\rceil$ by Proposition \ref{prop_minimum}.
\end{notation}

\begin{theoreme}\label{theorem_sufficient}
	Let $A$ and $B$ be configurations such that $|A|=|B|\eqdef N$ and:
	\begin{enumerate}[label=(\roman*),noitemsep,nolistsep]
		\item $\Span(A) \supseteq B$ is a single  $m \times n$ rectangle.
		\item $A$ has 2 extra coins.
		\item $B$ has 2 redundant coins.
		\item $N \geq \frac{3}{2}\max(\min_A,\min_B)+2$.
	\end{enumerate}
	Then $A \to B$.
\end{theoreme}

\noindent We would have liked conditions \textit{(i)} and \textit{(ii)} to be replaced by the sole condition that $A$ has 2 extra coins relatively to $B$, however we are not sure how the proof would work in that case. Apart from that, the additional assumption compared to Theorem \ref{theoreme_plusdeux} is condition \textit{(iv)}, which is not about the quality of the coins involved (extra/redundant) but purely about their quantity, as was suggested by the worst-case puzzles from Corollary \ref{coro_counterexample}. Moreover, these puzzles also show that the bound from condition \textit{(iv)} is almost tight: indeed, if $n$ is odd for instance, we have $\min_{A_n}=n$ and $\min_{B_n}=n+1$ so that the puzzle $A_n \xrightarrow{?} B_n$ satisfies $N=\frac{3n-3}{2}=\frac{3}{2}\max(\min_{A_n},\min_{B_n})-3$, just 5 coins away from this bound.

\vspace{1\baselineskip}
\noindent We now proceed with the proof of Theorem \ref{theorem_sufficient}. As usual, we are going to route through the canonical configurations, which means the challenge is to go from $L_A$ to $L_B$. The proof relies on an intuitive way to do so, which consists in forming a wave of coins (by flipping ‘L’s) to sweep across the board while dropping coins at all positions in $L_B$, as detailed in the proof of the following lemma. Note that this lemma is more general than we use, since the target configuration is not required to be canonical.

\begin{lemme}\label{lemma_sufficient}
	Let $m,n \geq 1$ and let $L$ be the canonical ‘L’ of size $m \times n$. Let $k \in \N$. If a configuration $C \subseteq \Span(L)$ satisfies $|C|<\min\left(\left\lceil\frac{m+n}{2}\right\rceil-\left\lceil\frac{\min(m,n)}{2}\right\rceil+(k-1),2(k-1)\right)$, then $L^{+k} \to C^{+k+|L|-|C|}$.
\end{lemme}

\begin{proof}

	We proceed by induction on the half-perimeter $m+n$. If $m=n=1$, then $C=L$ or $C=\varnothing$ so the result is obvious. Suppose $m+n \geq 3$ and assume the result holds for any half-perimeter lesser than $m+n$. Up to a 90 degree rotation of the board, also assume $m \leq n$. We divide $R \defeq \Span(L)$ into two rectangles: a bottom half $R_1$ and a top half $R_2$ (if $n$ is even then both halves are equal, otherwise we choose one of them arbitrarily to be bigger than the other by one row). Define $C_1 \defeq C \cap R_1$ and $C_2 \defeq C \cap R_2$. Up to swapping the roles of $R_1$ and $R_2$, assume $|C_1| \leq |C_2|$. Let $n' \in \left\{\left\lfloor\frac{n}{2}\right\rfloor,\left\lceil\frac{n}{2}\right\rceil\right\}$ be the number of rows of $R_2$, so that $R_2$ is of size $m \times n'$. To reach $C$, we build the bottom half $C_1$ first, then the top half $C_2$.
	
	\begin{enumerate}[label=\fbox{\arabic*}]
	
		\item We start by building the bottom half of $C$. We show that $L^{+k} \to C_0^{+k'}$ where $k' \defeq k+|L|-|C_0|$ and $C_0$ is the configuration defined as follows: $C_0 \cap R_1 = C_1$, and $C_0 \cap R_2=L_2$ is the canonical ‘L’ of size $m \times n'$. See Figure \ref{New_sufficient}.
			\begin{figure}[h]
				\centering
				\includegraphics[scale=.4]{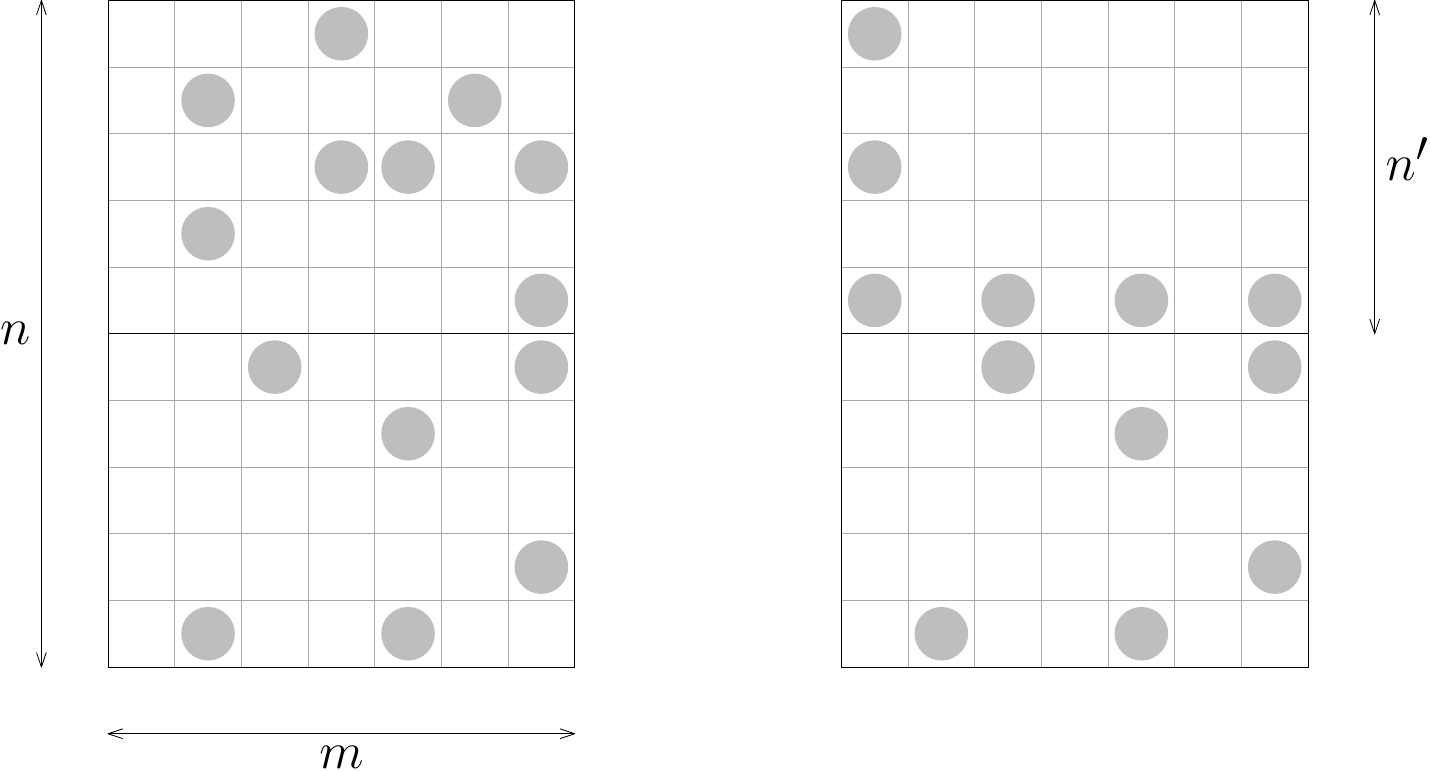}
				\caption{An example of a configuration $C$ (left) with its associated configuration $C_0$ (right).}\label{New_sufficient}
			\end{figure}
			\\ By our assumption on $C$, we have $k-1>\frac{|C|}{2} \geq |C_1|$ (recall that we have assumed without loss of generality that $|C_1| \leq |C_2|$ i.e. $C_1$ contains at most half of the coins in $C$) so $k \geq |C_1|+2$. This allows us to view our $k$ coins in hand as follows:
			\begin{itemize}[noitemsep,nolistsep]
				\item We have 2 \textit{supporting coins} that we will use to transform chains.
				\item We have $|C_1|$ \textit{building coins} that we will drop at the right positions to build $C_1$.
				\item If $k >|C_1|+2$, the remaining $k-|C_1|-2$ coins will be kept in hand.
			\end{itemize}
			Note that the supporting coins might not remain the same throughout the moves. For example, we might drop a supporting coin, perform some moves, and then pick up a coin: in that case, the picked up coin becomes a supporting coin even if it is not "physically" the same coin that we dropped initially.
			We proceed in four steps.
			
			\begin{enumerate}[label=(\alph*)]
			
				\item Let $p$ be the bottom-left corner of $R_2$: we want to make sure there is a coin at $p$.
					\begin{itemize}[noitemsep,nolistsep]
						\item[--] Case (a1): even $R$, odd $n'$. There already is a coin at $p$.
						\item[--] Case (a2): even $R$, even $n'$. We drop one of our supporting coins at $p$.
						\item[--] Case (a3): odd $R$. If needed, we use a leapfrog to put one of the two adjacent coins at $p$. Recall that a leapfrog only uses one supporting coin.
					\end{itemize}
					\begin{figure}[h]
						\centering
						\includegraphics[scale=.4]{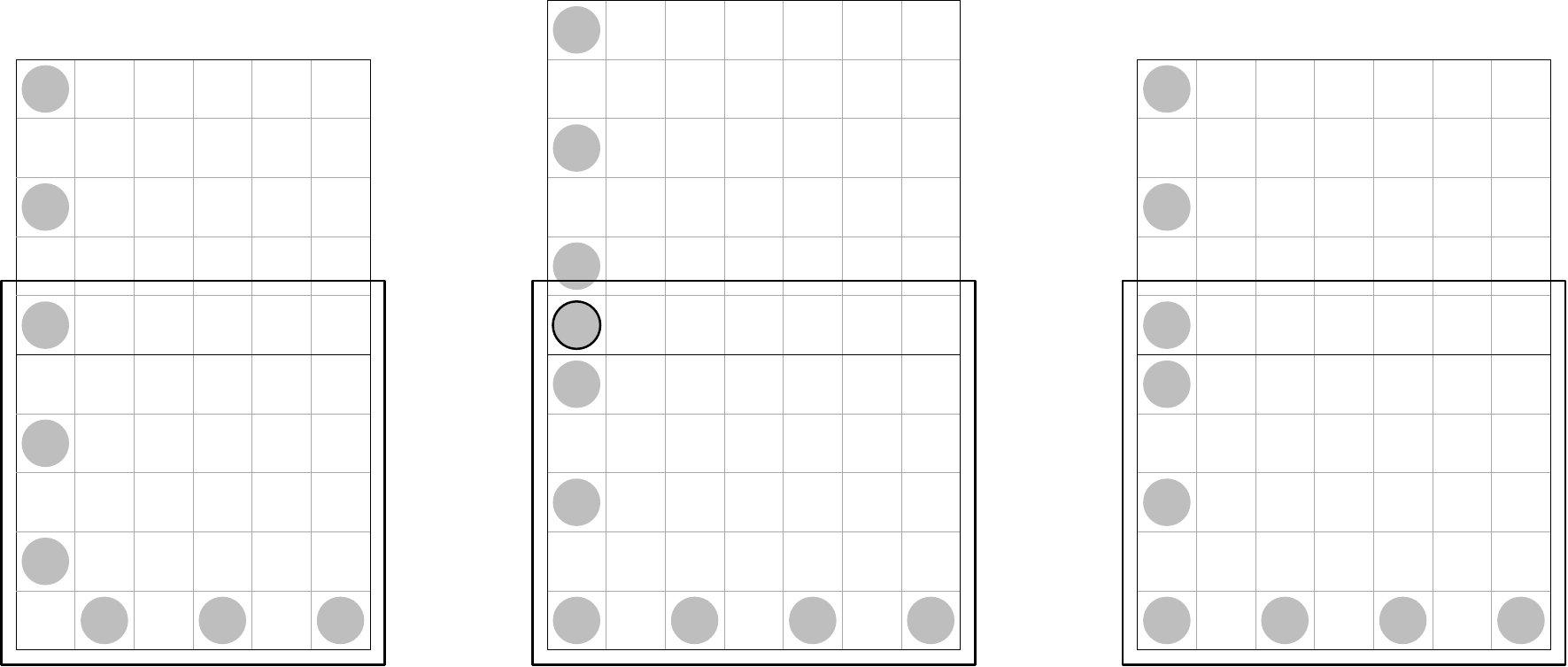}
						\caption{The board after step (a). From left to right: case (a1), case (a2), case (a3).}\label{New_sufficient2}
					\end{figure}
				
				\item The board now contains an ‘L’ whose extremal coins are at the bottom-left corner of $R_2$ and the bottom-right corner of $R_1$, as highlighted in Figure \ref{New_sufficient2}. We now use our supporting coin(s) to flip this ‘L’, while building $C_1$ in the process (see Figure \ref{New_sufficient3} for the desired result). To achieve this, we take advantage of the fact that the flip sweeps over the entirety of $R_1$. In the subroutines from Figures \ref{Flipping_even} and \ref{Flipping_odd}, some positions are highlighted by a black outline: whenever one of these positions contains a coin in $C_1$, we drop a building coin there at the appropriate moment during the subroutine (one example is detailed in Figure \ref{Building}). Over the flip as a whole, these positions cover all of $R_1$ except its rightmost column, so that all of $C_1$ is correctly replicated at the end of this step apart from the coins in the rightmost column. Note that the ‘L’ that we flip is odd in case (a2), so that the lone remaining supporting coin is indeed enough to flip it. See Figure \ref{New_sufficient3}.
					\begin{figure}[h]
						\centering
						\includegraphics[scale=.38]{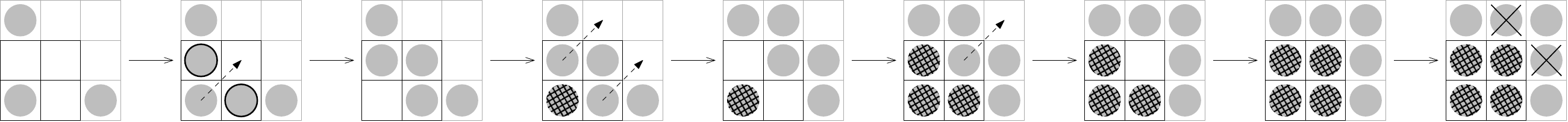}
						\caption{How to drop a building coin at any desired position (here we drop four of them, but we can drop less). The building coins are shaded.}\label{Building}
					\end{figure}
					\begin{figure}[h]
						\centering
						\includegraphics[scale=.4]{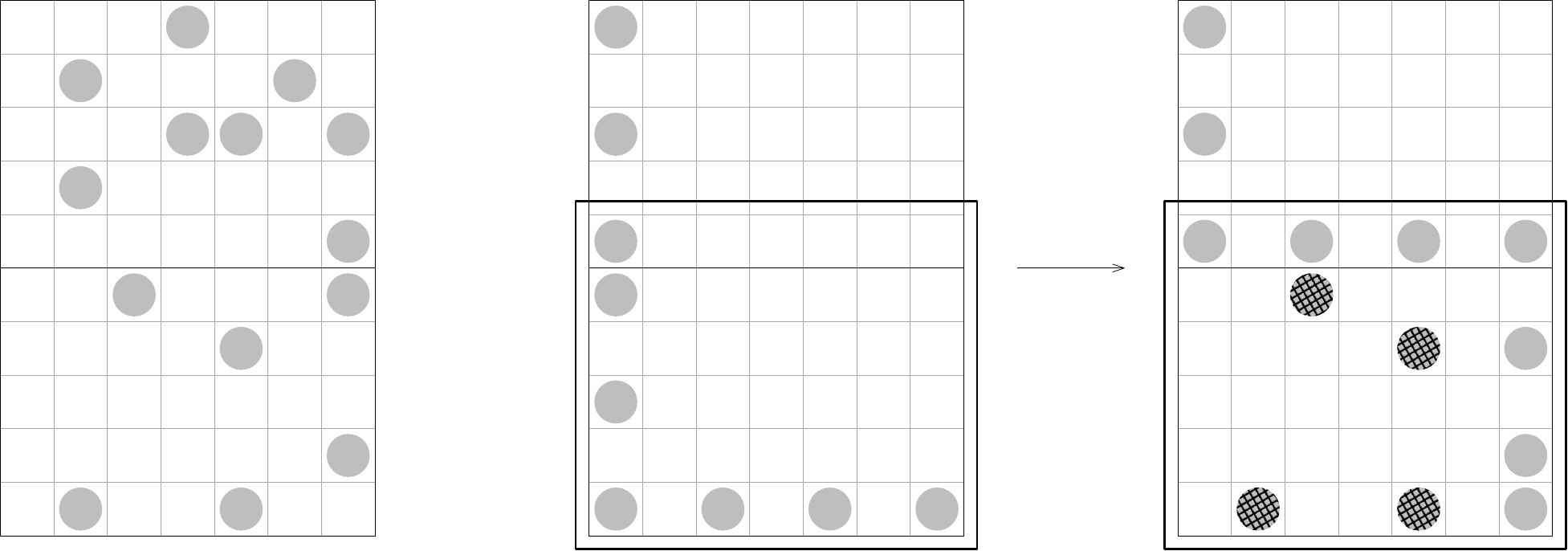}
						\caption{Left: a configuration $C$ (the same example as in Figure \ref{New_sufficient}). Middle: the board after step (a). Right: the board after step (b).}\label{New_sufficient3}
					\end{figure}
					
				\item We now make sure there is a coin in the bottom-right corner of $R_2$. If the newly flipped ‘L’ is odd, this is done with a leapfrog. If it is even, in particular we are not in case (a2), so we can afford to drop one of our two supporting coins at the bottom-right corner of $R_2$ if needed. In both cases, we still have at least one supporting coin at our disposal. We now correct the rightmost column of $R_1$: we drop building coins where they are needed in the holes in-between the coins that are already on board, and then we pick up all coins in that column that are not in $C_1$.
				
				\item At this point, $R_2$ either contains an ‘L’ or a chain that is almost an ‘L’ apart from the fact it has two pairs of adjacent coins (this can happen if we have dropped a supporting coin at the bottom-right corner of $R_2$ in step (c)). In this latter case, we use a leapfrog to retrieve the coin in excess so that $R_2$ contains a true ‘L’. Finally, we leapfrog if needed to make this ‘L’ canonical, so that the full configuration on board is now exactly $C_0$. These leapfrogs are always possible because we have at least one supporting coin at our disposal. See Figure \ref{New_sufficient4}.
					\begin{figure}[h]
						\centering
						\includegraphics[scale=.4]{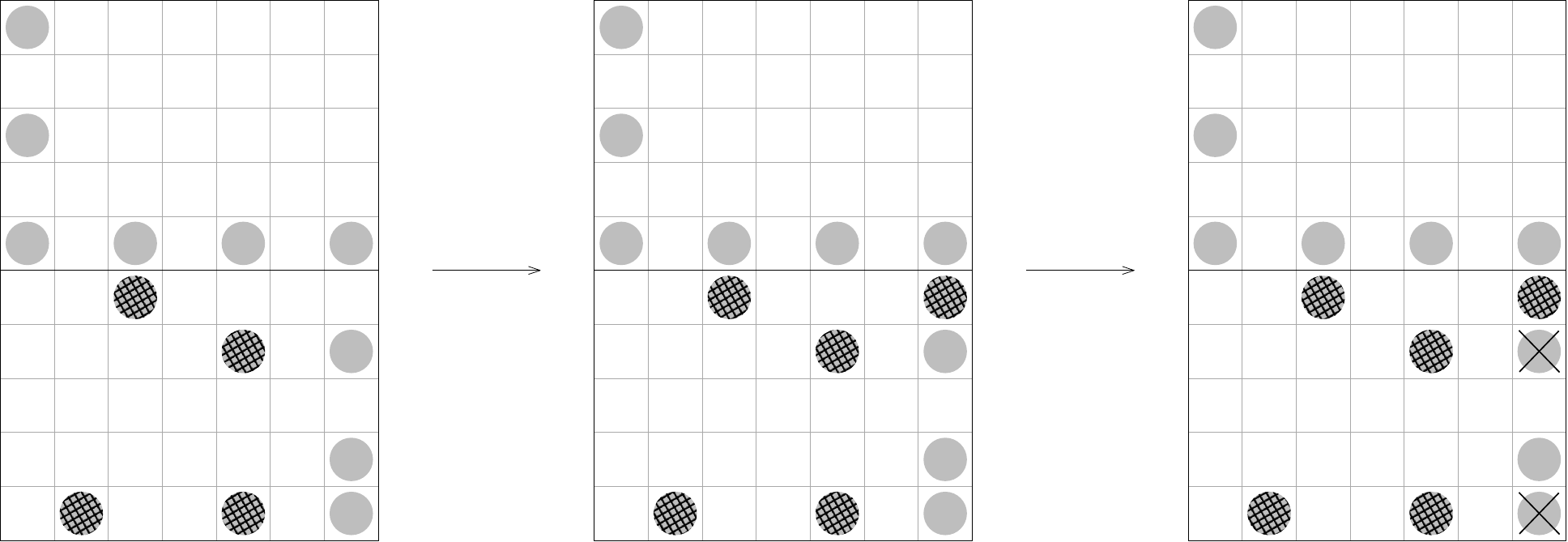}
						\caption{Step (d) performed as follow-up to Figure \ref{New_sufficient3}.}\label{New_sufficient4}
					\end{figure}
				
			\end{enumerate}
		
		\item We now build the top half of $C$. Since the bottom half is built already, we will not touch it, therefore we want to show that $L_2^{+k'} \to C_2^{+k'+|L_2|-|C_2|}$. By our induction hypothesis, since $R_2$ is of half-perimeter $m+n'<m+n$ (recall that $m+n \geq 3$, so $n \geq 2$ hence $n'<n$), it suffices to show that $|C_2|<\min\left(\left\lceil\frac{m+n'}{2}\right\rceil-\left\lceil\frac{\min(m,n')}{2}\right\rceil+(k'-1),2(k'-1)\right)$. We now show both inequalities. Beforehand, since the total number of coins at our disposal is $|L|+k$ and $|C_1|+|L_2|$ of them are on the board at this point, note that:
			\begin{equation}\label{equality} k' = (|L|+k)-(|C_1|+|L_2|) = \left( \left\lceil\frac{m+n}{2}\right\rceil + k \right) - \left( |C_1| + \left\lceil\frac{m+n'}{2}\right\rceil \right). \end{equation}
			\begin{itemize}
				\item By our assumption on $C$, we have $|C_1|=|C|-|C_2|<\left\lceil\frac{m+n}{2}\right\rceil-\left\lceil\frac{\min(m,n)}{2}\right\rceil+(k-1)-|C_2|$. Using equality (\ref{equality}), we get:
					\[ k' > |C_2| + 1 + \left\lceil\frac{\min(m,n)}{2}\right\rceil - \left\lceil\frac{m+n'}{2}\right\rceil \geq |C_2| + 1 + \left\lceil\frac{\min(m,n')}{2}\right\rceil - \left\lceil\frac{m+n'}{2}\right\rceil\,, \]
					where the last inequality comes from the fact that $n > n'$. We thus get the first desired inequality: $|C_2|<\left\lceil\frac{m+n'}{2}\right\rceil-\left\lceil\frac{\min(m,n')}{2}\right\rceil+(k'-1)$.
				\item By our assumption on $C$, we have $k-1>\frac{|C|}{2} \geq |C_1|$ hence $k-1- |C_1|>0$. Moreover $k-1-|C_1|=k-1-|C|+|C_2|>|C_2|-\left\lceil\frac{m+n}{2}\right\rceil+\left\lceil\frac{\min(m,n)}{2}\right\rceil = |C_2|-\left\lceil\frac{m+n}{2}\right\rceil+\left\lceil\frac{m}{2}\right\rceil$. Since all integers $x>y$ with $x>0$ satisfy $x\geq\frac{y}{2}+1$, we get $k-1-|C_1| \geq \frac{1}{2}\left(|C_2|-\left\lceil\frac{m+n}{2}\right\rceil+\left\lceil\frac{m}{2}\right\rceil\right)+1$.
					\\ Recall that $k'-1 = (k-1-|C_1|)+\left\lceil\frac{m+n}{2}\right\rceil-\left\lceil\frac{m+n'}{2}\right\rceil$ by equality (\ref{equality}). Therefore:
					\begin{align*}
						k'-1 & \geq \frac{|C_2|}{2} + \frac{1}{2}\left\lceil\frac{m+n}{2}\right\rceil + \frac{1}{2}\left\lceil\frac{m}{2}\right\rceil - \left\lceil\frac{m+n'}{2}\right\rceil + 1 \\
							 & \geq \frac{|C_2|}{2} + \frac{m+n}{4} + \frac{m}{4} - \frac{m+\frac{n+1}{2}+1}{2} + 1 = \frac{|C_2|}{2} + \frac{1}{4} \\
							 & > \frac{|C_2|}{2} \,,
					\end{align*}
					from which $|C_2|<2(k'-1)$ which concludes. \qedhere
			\end{itemize}
			
	\end{enumerate}
	\renewcommand{\qedsymbol}{}
\end{proof}

\begin{proof}[Proof of Theorem \ref{theorem_sufficient}]
	In fact, we prove a more general result where condition \textit{(iv)} is replaced by the following double inequality:
	\begin{align}
		N > & \,\, {\textstyle \min_A} + \frac{\min_B}{2} + 1 \label{inequality1} \\
		N > & \,\, {\textstyle \min_B} + \left\lceil\frac{\min(m,n)}{2}\right\rceil +1 \label{inequality2}
	\end{align}
	Let us first check that this assumption is indeed weaker. Suppose that \textit{(iv)} holds, then:
	\begin{itemize}[noitemsep]
		\item $N \geq \min_A + \frac{\min_B}{2} + 2$. Therefore, (\ref{inequality1}) holds.
		\item $N \geq \min_B + \frac{\min_A}{2} + 2$. Since $N$ is an integer, this yields $N \geq \min_B + \left\lceil\frac{\min_A}{2}\right\rceil + 2$. Moreover $\left\lceil\frac{\min_A}{2}\right\rceil = \left\lceil \frac{1}{2}\left\lceil\frac{m+n}{2}\right\rceil \right\rceil = \left\lceil\frac{m+n}{4}\right\rceil \geq \left\lceil\frac{\min(m,n)}{2}\right\rceil$, where we have used the fact that any real number $x$ satisfies $\left\lceil\frac{x}{2}\right\rceil = \left\lceil\frac{\lceil x \rceil}{2}\right\rceil$. Therefore, we get $N \geq \min_B + \left\lceil\frac{\min(m,n)}{2}\right\rceil + 2$, so (\ref{inequality2}) holds.
	\end{itemize}
	Assume that conditions \textit{(i)},\textit{(ii)},\textit{(iii)},(\ref{inequality1}),(\ref{inequality2}) all hold. As already mentioned, we want to use Corollary \ref{coro_main_lemma} (in this case $L_{A_0}=L_A$ because of our assumption that $A$ has 2 extra coins) so we need to show that $L_A^{+k} \to L_B^{+k+|L_A|-|L_B|}$ where $k \defeq |A|-|L_A|=N-\min_A=N-\left\lceil\frac{m+n}{2}\right\rceil$. By Lemma \ref{lemma_sufficient}, it suffices to show that $\min_B=|L_B|<\min\left(\left\lceil\frac{m+n}{2}\right\rceil-\left\lceil\frac{\min(m,n)}{2}\right\rceil+(k-1),2(k-1)\right)$.
	\begin{itemize}[noitemsep]
		\item By (\ref{inequality1}), we have $\min_B < 2(N-\min_A-1)=2(k-1)$.
		\item By (\ref{inequality2}), we have $\min_B < N - \left\lceil\frac{\min(m,n)}{2}\right\rceil - 1 = \left\lceil\frac{m+n}{2}\right\rceil-\left\lceil\frac{\min(m,n)}{2}\right\rceil+(k-1)$. \qedhere
	\end{itemize}
	\renewcommand{\qedsymbol}{}
\end{proof}

\noindent Even though, as we have mentioned, the bound in Theorem \ref{theorem_sufficient} is tight up to an additive factor $O(1)$ for a square span ($n=m$, as in $A_n$ and $B_n$), things might be different for a rectangular $m \times n$ span in general.

\vspace{1\baselineskip}
\section{One extra coin: minimum+1 configurations}\label{Section5}
\vspace{1\baselineskip}

\noindent The case where $A$ only has 1 extra coin relatively to $B$ is even more complicated. In this section, we initiate its study with the following restrictions:
\begin{itemize}[noitemsep,nolistsep]
	\item $\Span(A)=\Span(B)$ is a single rectangle.
	\item $A$ is \textit{minimum+1}, which means that $A$ consists of a minimum configuration with same span as $A$ plus one coin added to it. In particular, $A$ has 1 extra coin (and not more).
\end{itemize}
Note that under these restrictions, and further assuming that $B$ has 1 redundant coin $b$ which we know is necessary, $B$ is also minimum+1: indeed, $B \setminus \{b\}$ has same span as $A$ so it is minimum by cardinality.

\subsection{Even span}

\noindent The case of an even span is straightforward: if $A$ is a minimum+1 configuration with even span, then the only way to make moves from $A$ without decreasing the span is to move the same coin over and over again.

\begin{proposition}
	Let $A \neq B$ be minimum+1 configurations such that $\Span(A)=\Span(B)$ is an even rectangle. Then $A \to B$ if and only if $A \mapsto B$.
\end{proposition}

\begin{proof}
	Write $B=M \cup \{b\}$ where $M$ is a minimum configuration with same span as $B$. Since this span is an even rectangle, we know $M$ has all isolated coins by Proposition \ref{prop_minimum}, therefore $A \to B$ if and only if $A \mapsto B$ by Proposition \ref{trivial2}.
\end{proof}

\subsection{Odd span}

\noindent On the contrary, the slightly higher density of coins allows for some non-trivial movement when the span is odd. For instance, we have seen (and it was already noticed in \cite{DDV02}) that one coin in hand is enough to flip an odd ‘L’, so we know that the puzzle on the left of Figure \ref{Example_minimum} is solvable. In this example, the big difference compared to the even case is the presence of a pair of adjacent coins, which helps the moves and makes this "+1" situation more of a "+1.5" situation in a way. When it comes to a 90 degree \textit{rotation} of an odd ‘L’ however, even though it is easy to do with two coins in hand, it seems impossible with just one (as noticed again in \cite{DDV02}), which we will confirm shortly so that the puzzle on the right of Figure \ref{Example_minimum} is not solvable. The challenge is to figure out where the limit is between what is achievable and what is not, and we now bring our contribution in that direction. In what follows, we fix an odd $m \times n$ rectangle $R$.

\begin{figure}[h]
	\centering
	\includegraphics[scale=.4]{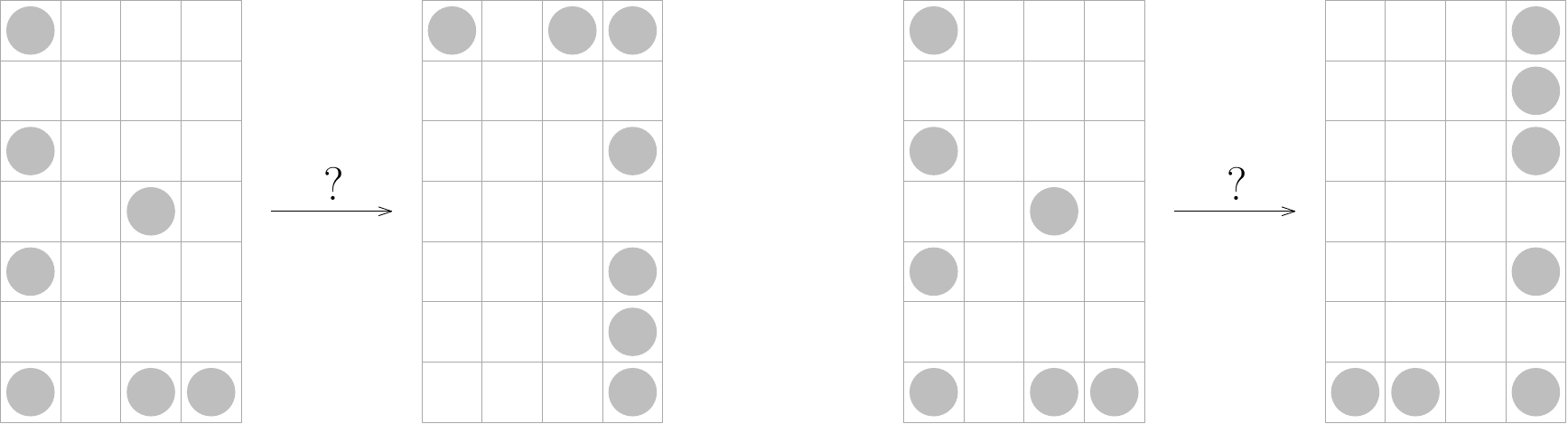}
	\caption{Two examples of puzzles involving minimum+1 configurations with even span. The one on the left is solvable, the one on the right is not.}\label{Example_minimum}
\end{figure}

\subsubsection{Reduction to the poking game}
	
\noindent First of all, we can make an interesting observation:

\begin{lemme}\label{lemma_pushing}
	Let $A$ be a minimum+1 configuration, and consider a sequence of moves $A=A_0 \xmapsto{c_1 \,\mapsto p_1} A_1 \xmapsto{c_2 \,\mapsto p_2} A_2 \mapsto \ldots$ such that the span is preserved throughout and no coin is moved twice in a row. Then:
	\begin{itemize}[noitemsep,nolistsep]
		\item For all $t \geq 1$, $A_t$ is minimum+1.
		\item For all $t \geq 1$, $c_{t+1}$ is a neighbor of $p_t$. In other words, the coin that we move (apart from the very first one) is always a neighbor of the most recently moved coin.
	\end{itemize}
\end{lemme}

\begin{proof}
	Let $t \geq 1$.
	\begin{itemize}[noitemsep,nolistsep]
		\item $\Span(A_t \setminus \{p_t\})=\Span(A_t)=\Span(A)$, therefore $A_t \setminus \{p_t\}$ is minimum by cardinality, so $A_t$ is minimum+1.
		\item Suppose for a contradiction that $c_{t+1}$ is not a neighbor of $p_t$. Since no coin in moved twice in a row, we know $p_t \neq c_{t+1}$, so this means that $p_t$ has at least two neighbors in $A_t \setminus \{c_{t+1}\}$ by the 2-adjacency rule, which yields $\Span(A_t \setminus \{c_{t+1}\})=\Span(A_t \setminus \{c_{t+1},p_t\})$. Finally, since $A_{t+1}\setminus\{p_{t+1}\}=A_t \setminus \{c_{t+1}\}$, we get:
	\[ \Span(A_t)=\Span(A_{t+1})=\Span(A_{t+1}\setminus\{p_{t+1}\})=\Span(A_t \setminus \{c_{t+1}\})=\Span(A_t \setminus \{c_{t+1},p_t\})\,,\]
	therefore $A_t$ has 2 extra coins which is a contradiction. \qedhere
	\end{itemize}
	\renewcommand{\qedsymbol}{}
\end{proof}

\noindent Conceptually, it is helpful to forget about the "+1" coin, by ignoring the last moved coin at all times. This allows us to focus on the underlying minimum configuration and how it evolves during the moves. Since each coin that we move is a neighbor of the most recently moved coin, as guaranteed by Lemma \ref{lemma_pushing}, each move causes the underlying minimum configuration to see one of its coins effectively slide by one position towards another coin. An example is given in Figure \ref{Pushing1} (the first move is ignored as it will be in the reduction to come). This inspires us to consider a new game, which is played directly on minimum configurations instead of minimum+1 configurations.

\begin{figure}[h]
	\centering
	\includegraphics[scale=.4]{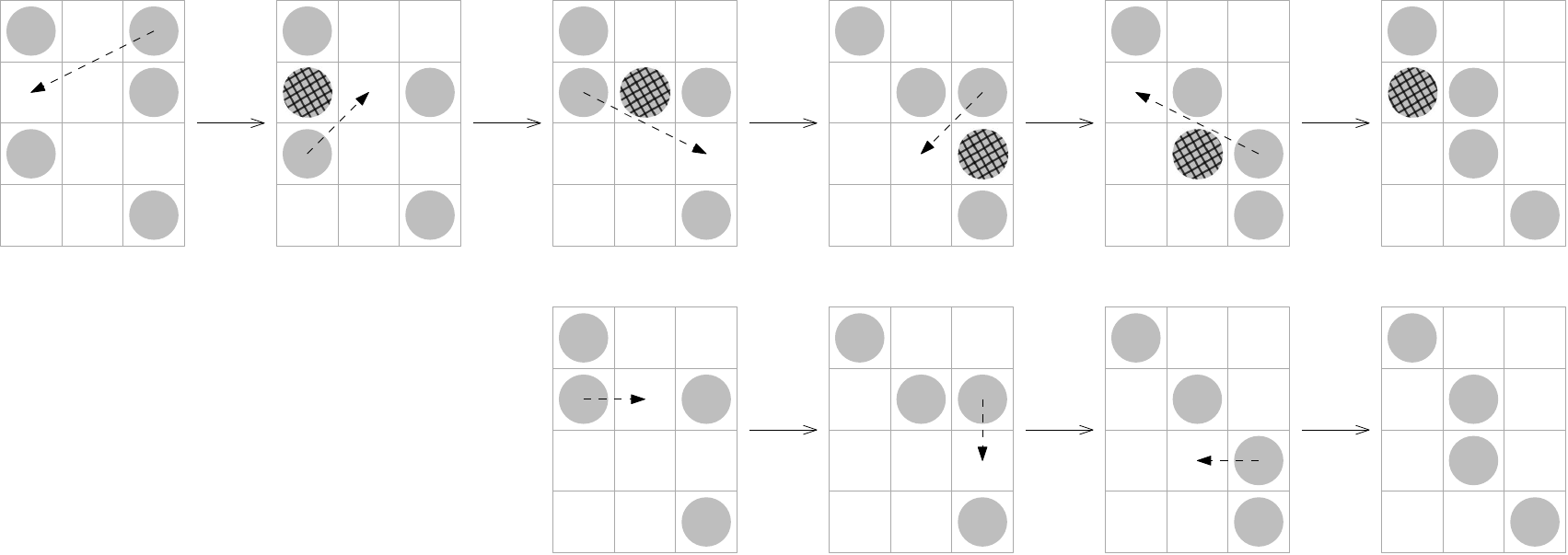}
	\caption{Top: a sequence of moves (the last moved coin is shaded). Bottom: evolution of the underlying minimum configuration obtained by deleting the last moved coin.}\label{Pushing1}
\end{figure}

\begin{definition}
	Let $M$ be a minimum configuration with a pair $\{c_1,c_2\}$ of adjacent coins (recall that, by Proposition \ref{prop_minimum}, $M$ contains at most one such pair). A \textit{poke} consists, for some $i \in \{1,2\}$, in sliding $c_i$ to one of its free neighboring positions $p$ with the restriction that $p$ must have at least one neighboring coin $c$ other than $c_i$ (we might say that $c_i$ is poked \textit{towards} $c$). This is denoted by $c_i \xmapsto{\P} p$.
\end{definition}

\begin{notation}
	Let $M$ and $M'$ be minimum configurations with span $R$, each containing a pair of adjacent coins.
	\begin{itemize}[noitemsep,nolistsep]
		\item If there is a single poke $c \xmapsto{\P} p$ from $M$ to $M'$, we write $M \xmapsto{c \,\xmapsto{\P} p} M'$ or simply $M \xmapsto{\P} M'$.
		\item If there exists a sequence of pokes from $M$ to $M'$, we write $M \xrightarrow{\P} M'$.
	\end{itemize}
\end{notation}

\begin{proposition}
	A poke is always reversible: if $M_0 \xmapsto{c \,\xmapsto{\P} p} M_1$ then $M_1 \xmapsto{p \,\xmapsto{\P} c} M_0$. In particular, for any two minimum configurations $M$ and $M'$ with span $R$ each containing a pair of adjacent coins, we have $M \xrightarrow{\P} M'$ if and only if $M' \xrightarrow{\P} M$.
\end{proposition}

\begin{proof}
	If $\{c,c_2\}$ denotes the pair of adjacent coins in $M_0$, then $p$ is part of the pair of adjacent coins in $M_1$ and $c_2$ is a neighboring coin of the free position $c$ in $M_1$ which makes $p \xmapsto{\P} c$ a valid poke from $M_1$.
\end{proof}

\begin{proposition}\label{properties_push}
	Any poke preserves the span, and the configuration obtained after a poke is still a minimum configuration with a pair of adjacent coins.
\end{proposition}

\begin{proof}
Let $M$ be a minimum configuration with a pair $\{c_1,c_2\}$ of adjacent coins. Let $M'$ be the configuration obtained from $M$ after some poke $c_1 \xmapsto{\P} p$ (up to swapping the roles of $c_1$ and $c_2$). By definition of a poke, $p$ has a neighboring coin $c \in M \setminus \{c_1\}$. Obviously $M'$ contains a pair of adjacent coins, namely $\{p,c\}$. To conclude, it suffices to show that $M'$ has same span as $M$, because this implies that $M'$ is minimum by cardinality. In $M$, the position $p$ is free but has two distinct occupied neighbors $c_1$ and $c$, therefore $\Span(M)=\Span(M \cup \{p\})$. In $M'$, the position $c_1$ is free but has two distinct occupied neighbors $c_2$ and $p$, therefore $\Span(M')=\Span(M' \cup \{c_1\})$. Since $M \cup \{p\}=M' \cup \{c_1\}$, we get  $\Span(M)=\Span(M')$.
\end{proof}

\vspace{1\baselineskip}
\noindent We have thus introduced a \textit{poking game} where, given two minimum configurations $M$ and $M'$ whose span is the same odd rectangle and both containing a pair of adjacent coins, we ask the question whether $M \xrightarrow{\P} M'$. The next result is a formal proof of the reduction illustrated in Figure \ref{Pushing1} and thus confirms that we can now focus entirely on the poking game.

\begin{proposition}\label{prop_pushing}
	Let $A$ and $B$ be minimum+1 configurations with span $R$. Suppose $A \neq B$ and $A \not\mapsto B$. Then $A \to B$ if and only if there exists a move $A \xmapsto{c_1 \,\mapsto p_1} A_1$ and $a \in A_1$, $b \in B$ such that:
	\begin{itemize}[noitemsep,nolistsep]
		\item $a$ is a neighbor of $p_1$.
		\item $b$ is a redundant coin in $B$.
		\item $A_1 \setminus \{a\} \xrightarrow{\P} B \setminus \{b\}$.
	\end{itemize}
	As a consequence, the case where $A$ and $B$ are minimum+1 configurations whose span is the same odd rectangle reduces to the poking game, up to a factor $O(N^2)$ where $N \defeq |A|=|B|$.
\end{proposition}

\begin{proof}
	First suppose $A \to B$. We write $A=A_0 \xmapsto{c_1 \,\mapsto p_1} A_1 \xmapsto{c_2 \,\mapsto p_2} \ldots \xmapsto{c_T \,\mapsto p_T} A_T = B$ where no coin is moved twice in a row. We know $T \geq 2$ because $A \neq B$ and $A \not\mapsto B$. Setting $a=c_2$ and $b=p_T$, since $c_2$ is a neighbor of $p_1$ by Lemma \ref{lemma_pushing} and $p_T$ is a redundant coin in $B$ by the 2-adjacency rule, it suffices to show that $A_1 \setminus \{c_2\} \xrightarrow{\P} B \setminus \{p_T\}$. For $2 \leq t \leq T$, we define $M_t \defeq A_t \setminus \{p_t\}$, the intermediary configuration obtained when the moved coin from $A_{t-1}$ to $A_t$ is "in the air": $M_t$ is minimum by cardinality because $\Span(M_t)=\Span(A_t)$. Note that $M_2=A_1 \setminus \{c_2\}$ and $M_T=B \setminus \{p_T\}$, so our goal is to get $M_2 \xrightarrow{\P} M_T$. See Figure \ref{Pushing2}, which completes Figure \ref{Pushing1} and illustrates the full reduction to the poking game.
	\begin{itemize}
		\item First of all, we show that $M_2,\ldots,M_T$ all contain a pair of adjacent coins. Let $2 \leq t \leq T$, we have $M_t=A_{t-1} \setminus \{c_t\}$. Since $p_{t-1} \in A_{t-1}$, we know $p_{t-1} \in M_t$: indeed, no coin is moved twice in a row hence $p_{t-1} \neq c_t$. Moreover $p_{t-1}$ has two neighbors in $A_{t-1}$ by the 2-adjacency rule, so it has one neighbor in $M_t$.
		\item Let $2 \leq t \leq T-1$, we check that $c_{t+1} \xmapsto{\P} p_t$ is a valid poke from $M_t$ to $M_{t+1}$.
			\begin{itemize}[noitemsep,nolistsep]
				\item[--] We have $(M_t \setminus \{c_{t+1}\}) \cup \{p_t\}=A_t \setminus \{c_{t+1}\}=A_{t+1} \setminus \{p_{t+1}\}=M_{t+1}$ so the poke $c_{t+1} \xmapsto{\P} p_t$ (if valid) does transform $M_t$ into $M_{t+1}$.
				\item[--] We know $c_{t+1} \in A_t$, moreover $c_{t+1} \neq p_t$ since no coin is moved twice in a row, so $c_{t+1} \in M_t$.
				\item[--] Obviously $p_t \not\in M_t$ by definition of $M_t$.
				\item[--] By Lemma \ref{lemma_pushing}, $c_{t+1}$ is a neighbor of $p_t$.
				\item[--] It only remains to check that $c_{t+1}$ is part of the pair of adjacent coins in $M_t$. But if it were not the case, then the (illegal) poke would create a second pair of adjacent coins i.e. $M_{t+1}$ would contradict Proposition \ref{prop_minimum}.
			\end{itemize}
	\end{itemize}
	In conclusion, we have $M_2 \xmapsto{c_3 \,\xmapsto{\P} p_2} M_3 \xmapsto{c_4 \,\xmapsto{\P} p_3} \ldots \xmapsto{c_T \,\xmapsto{\P} p_{T-1}} M_T$. The proof of the converse is similar: given the first move $A \xmapsto{c_1 \,\mapsto p_1} A_1$ and a sequence of pokes $(c_t \xmapsto{\P} p_t)_{2 \leq t \leq T}$ from $A_1 \setminus \{a\}$ to $B \setminus \{b\}$ for some suitable $a$ and $b$, we check that the sequence of moves $A \xmapsto{c_1 \,\mapsto p_1} A_1 \xmapsto{a \,\mapsto p_2} A_2 \xmapsto{c_2 \,\mapsto p_3} A_3 \xmapsto{c_3 \,\mapsto p_4} \ldots \xmapsto{c_{T-1} \,\mapsto p_T} A_T \xmapsto{c_T \,\mapsto b} B$ is valid. As for the last statement of this proposition, it comes from the fact there are: $O(N^2)$ possibilities for the first move, at most four possibilities for $a$ (since $p_1$ has at most four neighbors), and at most three possibilities for $b$ (indeed, given a redundant coin $b_0$ in $B$, the only coins other than $b_0$ that might be individually redundant in $B$ are those from the pair of adjacent coins in the minimum configuration $B \setminus \{b_0\}$, because the others are isolated in $B \setminus \{b_0\}$ and thus have at most one neighbor in $B$).
\end{proof}

\begin{figure}[h]
	\centering
	\includegraphics[scale=.4]{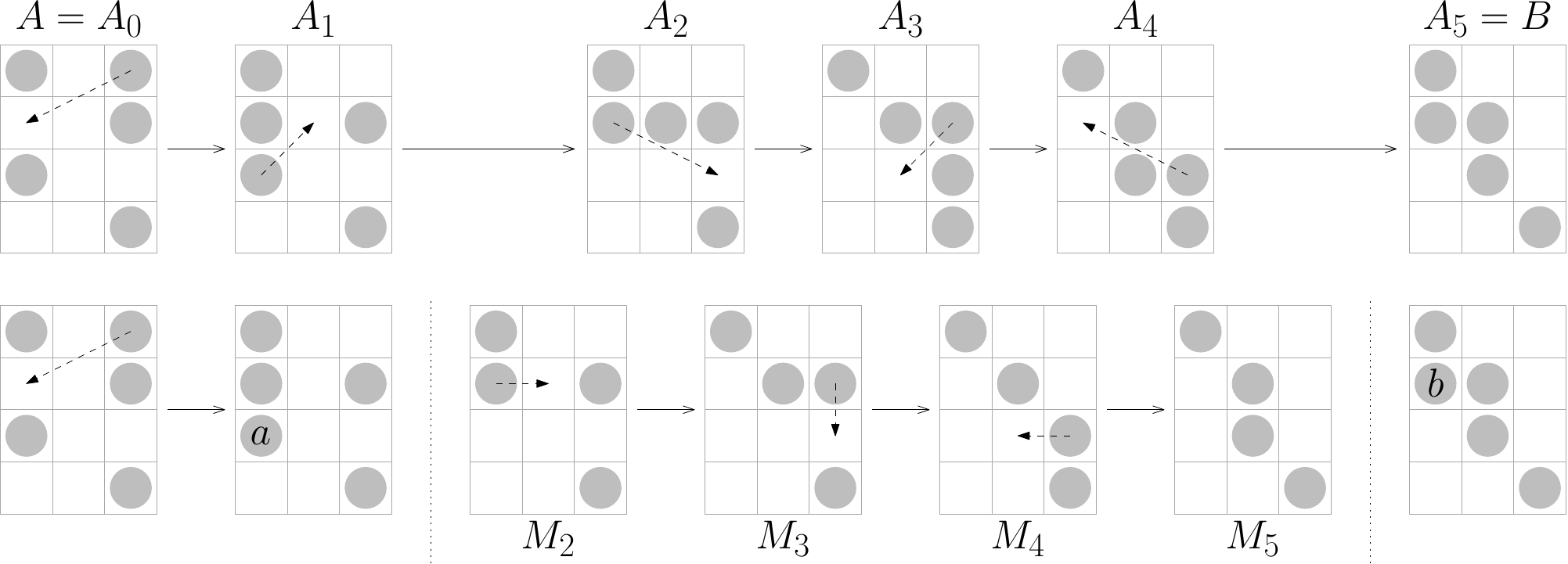}
	\caption{Top: a sequence of moves. Bottom: the corresponding reduction to the poking game, obtained by making the first move, deleting $a$, making a sequence of pokes, and then adding $b$.}\label{Pushing2}
\end{figure}

\subsubsection{The poking game on minimum chains}

It is natural to study the simplest minimum configurations first, namely minimum chains. We start by detailing their structure: a minimum chain with span $R$ containing two adjacent coins goes from one corner of $R$ to its opposite without taking any detour (i.e. in a "straight line").

\begin{proposition}\label{prop_chain}
	Let $M=[c_1,\ldots,c_N]$ be a minimum chain with span $R$ containing a pair of adjacent coins $\{c_{i_0},c_{i_0+1}\}$ for some $1 \leq i_0 \leq N-1$. Then:
	\begin{itemize}[noitemsep,nolistsep]
		\item For all $i \in \{1,\ldots,N-1\}$: $\dist(c_i,c_{i+1})=2$ if $i \neq i_0$ and $\dist(c_i,c_{i+1})=1$ if $i=i_0$.
		\item $c_1$ and $c_N$ are opposite corners of $R$.
		\item If $c_1$ is the top-left corner of $R$, then for all $i \in \{1,\ldots,N-1\}$, $c_{i+1}$ is either to the right, to the bottom, or to the bottom-right of $c_i$. The analogous assertion holds if $c_1$ is any other corner of $R$.
	\end{itemize}
\end{proposition}

\begin{proof} 
	Following the chain from $c_1$ to $c_N$, we consider the number of jumps that are made in each of the four directions. For instance, if for some $i \in \{1,\ldots,N-1\}$ the coin $c_{i+1}$ is to the bottom-right of $c_i$, then the transition from $c_i$ to $c_{i+1}$ counts as one right jump plus one bottom jump, whereas if $c_{i+1}$ is two positions under $c_i$ then this transition counts as two bottom jumps. The first assertion of this proposition follows directly from the definition of a chain and Proposition \ref{prop_minimum}, and further ensures that there are $2(N-2)+1=m+n-2$ jumps in total (indeed, recall that $N=\frac{m+n+1}{2}$ by Proposition \ref{prop_minimum}). By Proposition \ref{span_carre2}, there exist $t,b,l,r \in \{1,\ldots,N\}$ such that $c_t$ (resp. $c_b$, resp. $c_l$, resp. $c_r$) is in the top row (resp. bottom row, resp. leftmost column, resp. rightmost column) of $R$. Without loss of generality, assume $t \leq b$ and $l \leq r$, then at least $n-1$ bottom jumps are made from $c_t$ to $c_b$ and at least $m-1$ right jumps are made from $c_l$ to $c_r$. Since there are $m+n-2$ jumps in total, we conclude that: exactly $n-1$ bottom jumps are made from $c_t$ to $c_b$, exactly $m-1$ right jumps are made from $c_l$ to $c_r$, and those cover all jumps made, hence the third assertion. Since no top jump or left jump is ever made, in particular there is no top jump between $c_1$ and $c_t$ and no left jump between $c_1$ and $c_l$, therefore $c_1$ is the top-left corner of $R$. For the same reason, $c_N$ is the bottom-right corner of $R$, which proves the second assertion.
\end{proof}

\noindent Recall that one coin in hand is enough to flip an odd ‘L’. The following theorem, which generalizes this observation, is the central result of this section. In particular, it gives a formal proof that rotating an odd ‘L’ is impossible with a single coin in hand.

\begin{theoreme}\label{theorem_chains}
	If $M$ is a minimum chain with span $R$ containing a pair of adjacent coins, then $M \xrightarrow{\P} M'$ if and only if $M'$ is a minimum chain between the same two coins as $M$.
\end{theoreme}

\begin{proof}
	Write $M=[c_1,\ldots,c_N]$. We first show the necessity, and then we give a constructive proof of the sufficiency.
	\begin{itemize}
	
		\item Let $i_0 \in \{1,\ldots,N-1\}$ such that the pair of adjacent coins in $M$ is $\{c_{i_0},c_{i_0+1}\}$. We must show that, after a poke from $M$, we still obtain a chain between $c_1$ and $c_N$. Without loss of generality, assume that we poke $c_{i_0}$ onto some position $p$: we want to show that $[c_1,\ldots,c_{i_0-1},p,c_{i_0+1},\ldots,c_N]$ is a proper chain.
			\begin{itemize}[noitemsep,nolistsep]
				\item[--] We already know that, for all $i \in \{1,\ldots,N-1\}\setminus\{i_0-1,i_0\}$: $\dist(c_i,c_{i+1})\in\{1,2\}$.
				\item[--] Since $p$ is a neighbor of $c_{i_0}$ other than $c_{i_0+1}$, we have $\dist(p,c_{i_0+1})=2$.
				\item[--]  A consequence of the third assertion in Proposition \ref{prop_chain} is that, for any $i \in \{1,\ldots,N-1\}\setminus\{i_0\}$: $\dist(c_{i_0},c_i)=2(i_0-i)$ if $i<i_0$ and $\dist(c_{i_0},c_i)=2(i-i_0)-1$ if $i>i_0$. In particular, necessarily $i_0 \geq 2$ and $c_{i_0}$ is poked towards $c_{i_0-1}$, because it is the only coin at distance exactly 2 from $c_{i_0}$. This means $\dist(p,c_{i_0-1})=1$.
			\end{itemize}
			In conclusion, the configuration obtained after the poke is $[c_1,\ldots,c_{i_0-1},p,c_{i_0+1},\ldots,c_N]$, which is a chain between $c_1$ and $c_N$ since $i_0 \not\in\{1,N\}$. Moreover, it is minimum by cardinality, because the poke preserves the span according to Proposition \ref{properties_push}.
			
		\item Suppose $M'$ also forms a chain between $c_1$ and $c_N$, we want to show that $M \xrightarrow{\P} M'$. Without loss of generality, assume that $c=c_1$ is the top-left corner of $\Span(M)$. Let $L$ be the ‘L’ with end-points $c$ and $c'$ that hugs the left and bottom sides of $R$ and where the two adjacent coins are at the top-left. Since all pokes are reversible, it suffices to show that we can reach $L$: indeed, we would get $M \xleftrightarrow{\P} L \xleftrightarrow{\P} M'$. We prove this by induction on $|M|$. The case $|M|=2$ is trivial. Now assume $|M| \geq 3$ and suppose the result holds for chains of cardinality lesser than $|M|$.
			\begin{enumerate}[label={\arabic*.}]		
				\item Up to poking $c_2$ towards $c_3$ in the case where the pair of adjacent coins is $\{c_1,c_2\}$, we can assume that the chain $[c_2,\ldots,c_N]$ contains two adjacent coins.
				\item We can now apply the induction hypothesis on the chain $[c_2,\ldots,c_N]$.
				\item To finish, there are three possibilities depending on the direction of the jump from $c_1$ to $c_2$. If $c_2$ is to the bottom of $c_1$, then $c_1$ prolongs the ‘L’ we obtained by induction, so we just have to poke $c_2$ up towards $c_1$ to obtain the desired ‘L’. If $c_2$ is to the bottom-right (resp. right) of $c_1$, then we proceed as in Figure \ref{Chains2} (resp. Figure \ref{Chains4}). \qedhere
					\begin{figure}[h]
						\centering
						\includegraphics[scale=.393]{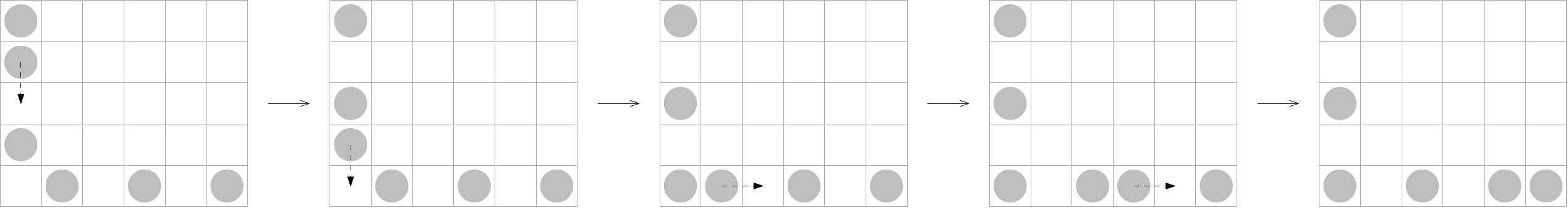}
						\caption{The poking version of the leapfrog.}\label{Leapfrog_pushing}
					\end{figure}
					\begin{figure}[h]
						\centering
						\includegraphics[scale=.4]{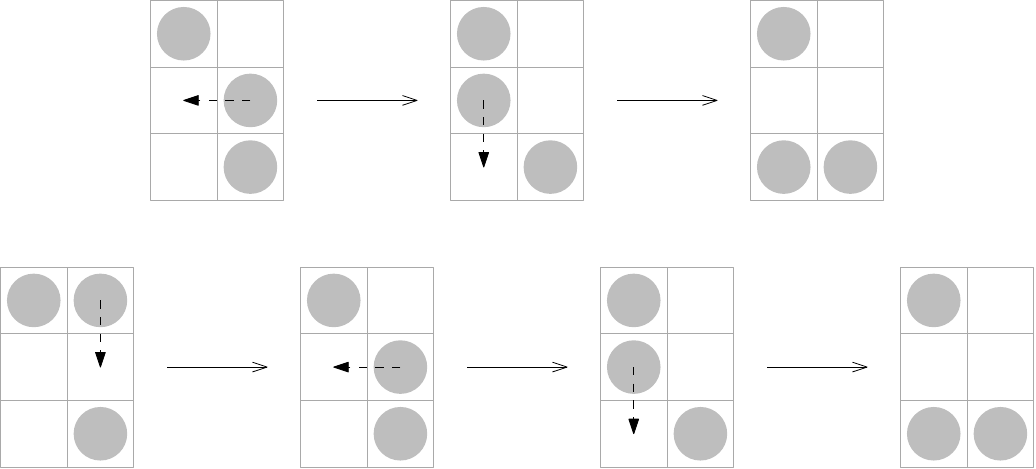}
						\caption{Subroutines used when $c_2$ is to the bottom-right of $c_1$.}\label{Chains1}
					\end{figure}
					\begin{figure}[h]
						\centering
						\includegraphics[scale=.393]{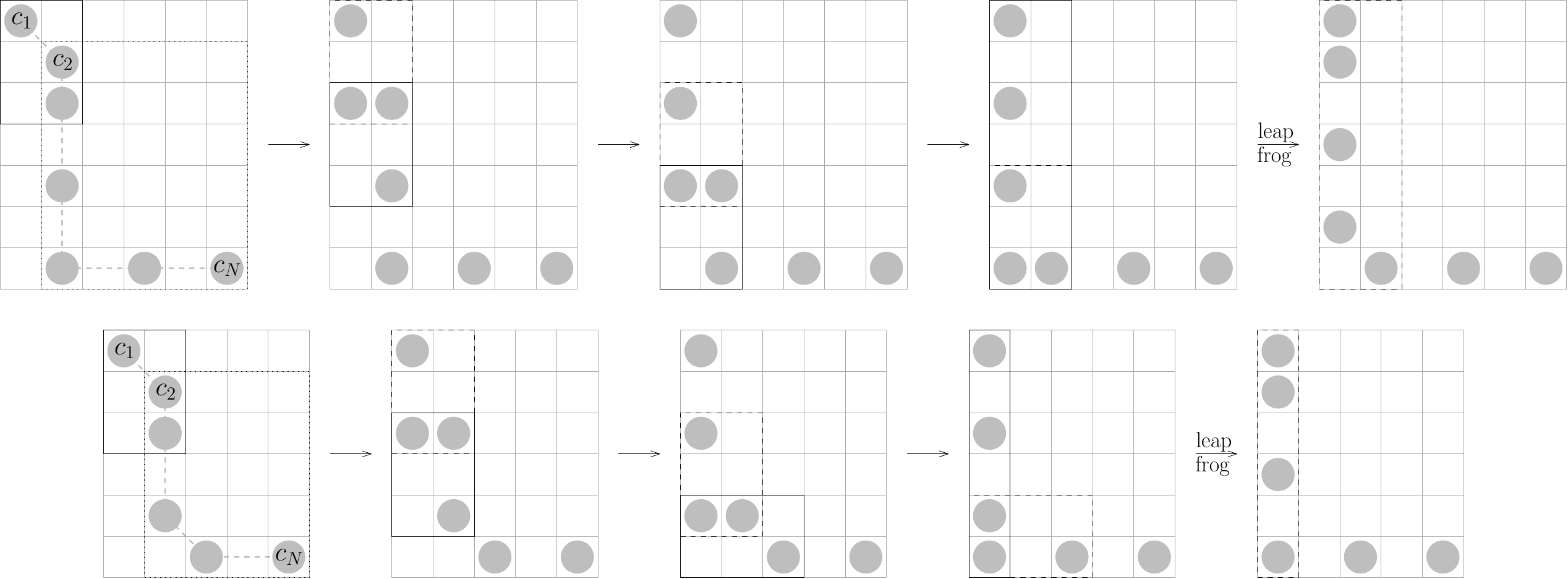}
						\caption{Case where $c_2$ is to the bottom-right of $c_1$, using the subroutines from Figures \ref{Leapfrog_pushing} and \ref{Chains1}. Top: even $m$, odd $n$. Bottom: odd $m$, even $n$.}\label{Chains2}
					\end{figure}
					\begin{figure}[h]
						\centering
						\includegraphics[scale=.4]{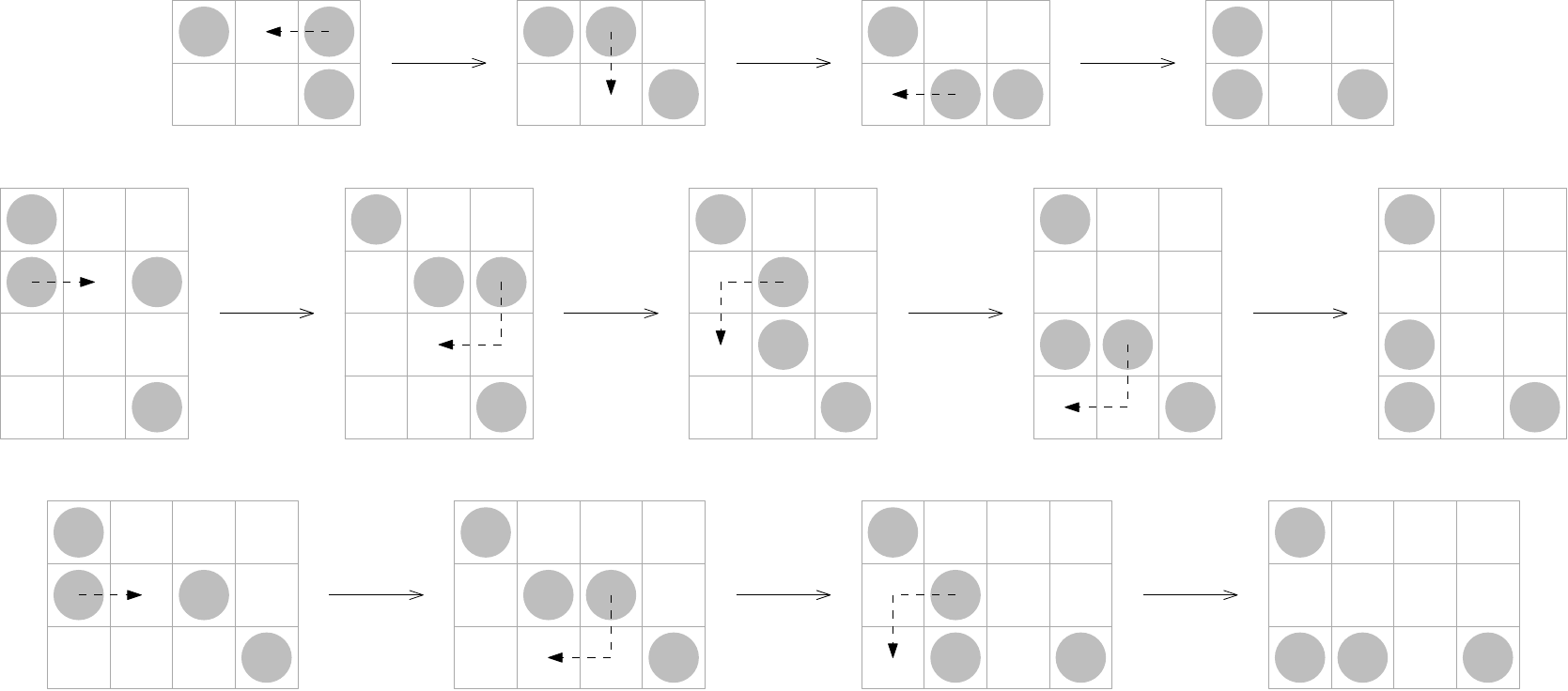}
						\caption{Subroutines used when $c_2$ is to the right of $c_1$.}\label{Chains3}
					\end{figure}
					\begin{figure}[h]
						\centering
						\includegraphics[scale=.393]{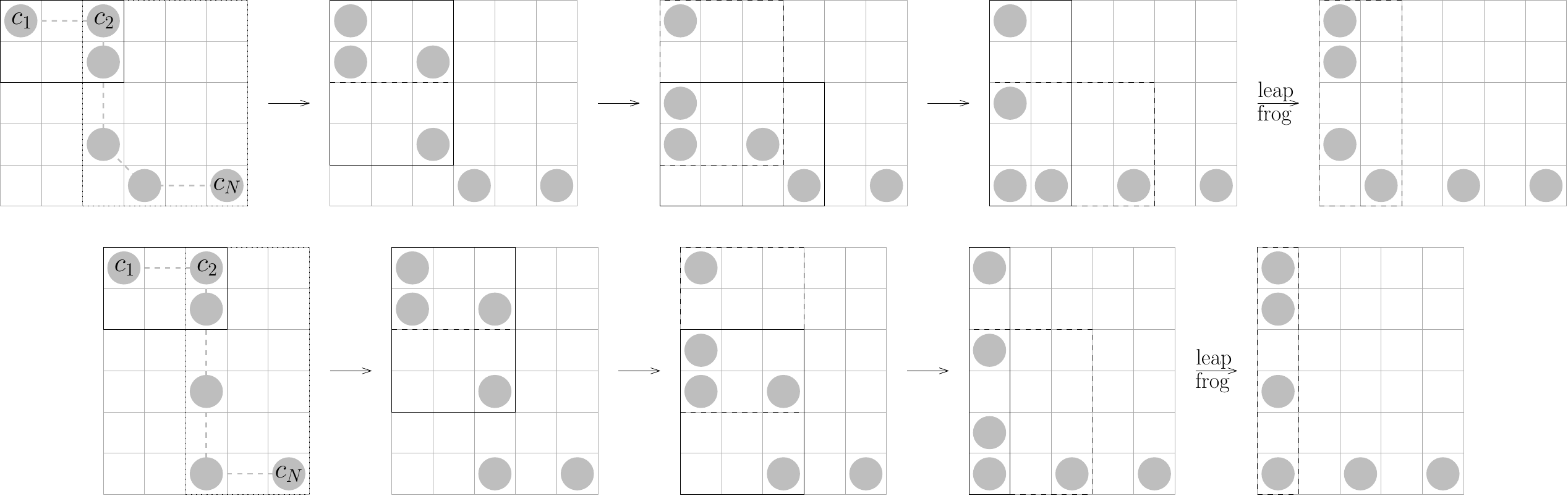}
						\caption{Case where $c_2$ is to the right of $c_1$, using the subroutines from Figures \ref{Leapfrog_pushing} and \ref{Chains3}. Top: even $m$, odd $n$. Bottom: odd $m$, even $n$.}\label{Chains4}
					\end{figure}
			\end{enumerate}
	\end{itemize}
	\renewcommand{\qedsymbol}{}
\end{proof}

\subsubsection{The poking game in general}

\noindent We have just seen that, if the starting configuration $M$ is a chain, then its end-points never move and all we can do is twist the string of coins that connects them. When playing the poking game on more general configurations, it quickly becomes apparent that certain coins are ever immobile and all we do is perform successive transformations of chains between some of them as in Theorem \ref{theorem_chains}:
\begin{itemize}[noitemsep,nolistsep]
	\item[--] We start by transforming a certain chain inside a certain rectangle $R_1$, until one of the two adjacent coins is at distance 2 from the first coin of a chain $C$ coming out of $R_1$.
	\item[--] We can then poke towards this coin and get a new chain, inside a new rectangle $R_2$, with $C$ prolonging one half of the previous chain.
	\item[--] We then transform this new chain, etc.
\end{itemize}
This is illustrated in Figure \ref{Illustrate_pushing}: in this example, all the action seems to happen inside three specific rectangles ($R_1$, $R_2$, $R_3$), with no possibility other than transitioning from one to the other via consecutive transformations of chains.

\begin{figure}[h]
	\centering
	\includegraphics[scale=.325]{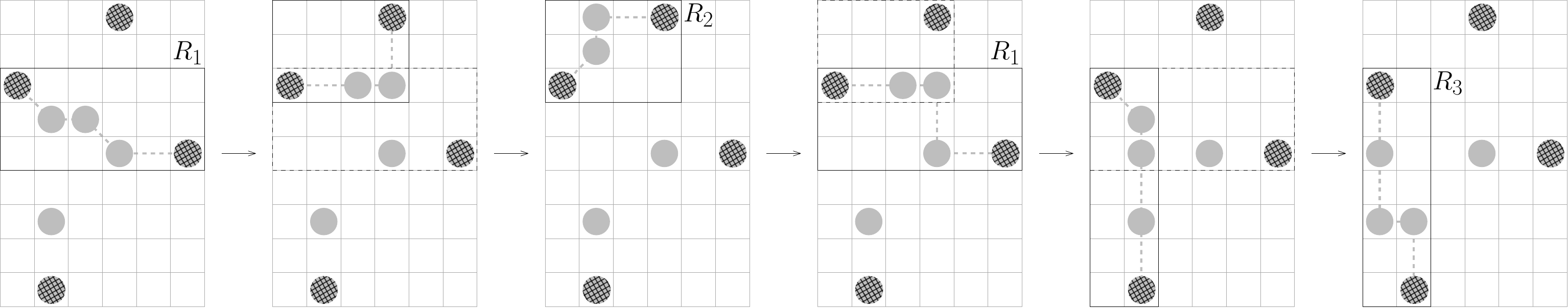}
	\caption{Playing the poking game. The outlined rectangle shows the chain that is currently being transformed. The shaded coins seem unmovable.}\label{Illustrate_pushing}
\end{figure}

\noindent An incomplete study of the poking game in general is made in \cite{Gal19}, the results from which can be summarized as follows:
\begin{itemize}
	\item First of all, Proposition \ref{prop_chain} can be generalized so as to get a full understanding of the structure of minimum configurations with a pair of adjacent coins. In particular, further exploiting of the perimeter considerations glimpsed in the proof of Proposition \ref{prop_minimum} show that these configurations are forests in terms of the \textit{2-connectivity} (connectivity at distance 1 or 2).
	\item A polynomial-time algorithm returns, given a configuration $M$, a specific set $F(M)$ of coins in $M$ (in Figure \ref{Illustrate_pushing}, these are the shaded coins). These coins are forever immobile, and all that is possible is to perform successive transformations of chains between some coins in $F(M)$. The set $F(\cdot)$ is invariant throughout the pokes.
	\item (If $M$ is a chain, recall that we already have Theorem \ref{theorem_chains}.)
	\item If $M$ is a tree in terms of the 2-connectivity, then $M \xrightarrow{\P} M'$ if and only if $F(M)=F(M')$.
	\item If $M$ is a forest in terms of the 2-connectivity (general case), then this does not work anymore. We conjecture that a refined algorithm provides us with an invariant that characterizes solvable puzzles, similarly to the case of trees.
\end{itemize}

\vspace{1\baselineskip}
\section*{Conclusion}
\vspace{1\baselineskip}

\hphantom{\indent}We have made some contributions to the study of the coin-moving game initiated in \cite{DDV02}. Our focus has been on configurations with either $2$ extra coins or $1$ extra coin, as suggested by the authors. Indeed, we use their canonicalization technique, which effectively nullifies the extra coins beyond the first two. Stepping outside of this framework however, more extra coins might be useful, so finding necessary and sufficient conditions in the case of 3 or more extra coins remains an open problem. For the same reason, we have not explored the case of 2 extra coins and just 1 redundant coin.
\medskip
\\ \indent In the case where there are 2 extra coins in $A$ and 2 redundant coins in $B$, Theorem \ref{theorem_sufficient} gives a sufficient condition that involves the total number of coins and the perimeters of the spans of $A$ and $B$. The worst-case puzzles described in Corollary \ref{coro_counterexample} show that the bound from the theorem is tight in the case of a square span, up to an additive factor $O(1)$.
\\ \indent However, for a rectangular $m \times n$ span in general, it looks like discrepancy between $m$ and $n$ tends to make puzzles easier. For instance, and provided we have the inclusion of spans, if $n=1$ then it seems easy to show that 1 extra coin and 1 redundant coin are always sufficient, and if $n=2$ then we suspect that $2$ extra coins and $2$ redundant coins are always sufficient. In general, we think that it could be possible to improve Theorem \ref{theorem_sufficient} by using only $\frac{\min(m,n)}{2}+O(1)$ additional coins.
\\ \indent We now know of two methods to solve general puzzles with 2 extra coins in $A$ and 2 redundant coins in $B$: the one from Theorem \ref{theoreme_plusdeux} consists in going from $A$ to ‘L’s then reverse into $B$, while the one from Theorem \ref{theorem_sufficient} consists in going from $A$ to an ‘L’ then sweep across the board and drop coins to create $B$. It would be good to design other solving methods that would apply to some puzzles that do not meet the conditions of either theorem, such as the puzzle on the left of Figure \ref{problem}.
\medskip
\\ \indent As for the case of a lone extra coin, we have focused on the case where $A$ and $B$ are minimum+1 configurations with the same span, for which the game can be reduced to a poking game. The case of chains is completely solved. A deeper study is made in \cite{Gal19}: a detailed description of the structure of minimum configurations is given, and allows for further results as well as leads towards a general solution. On the contrary, nothing is known about the structure of minimal configurations, so the general \textit{minimal+1} case remains obscure. Nevertheless, it is interesting to note that the poking game can be extended to that case (see \cite{Gal19} for the details). 
\medskip
\\ \indent Finally, it is easy to check that all proofs of our sufficiency results are constructive and provide polynomial-time algorithms to find a winning sequence of moves of polynomial length. However, the complexity of the game as a whole remains unknown.

\vspace{1\baselineskip}
\section*{Acknowledgments}

We would like to thank the reviewers for providing us with helpful feedback and inspiring the abstract of this paper.

\vspace{1\baselineskip}

\end{document}